\documentclass[lettersize,journal]{IEEEtran}

\usepackage{color}
\usepackage{cite}
\usepackage{amsmath,amssymb,amsfonts}
\usepackage{graphicx}
\usepackage{subcaption}
\usepackage{algorithm}
\usepackage{algpseudocode}
\usepackage{hyperref}
\hypersetup{hidelinks=true}
\usepackage{textcomp}

\usepackage{amsthm}
\theoremstyle{definition}
\newtheorem{definition}{Definition}
\newtheorem{lemma}{Lemma}

\newtheorem{theorem}{Theorem}

\newtheorem{assumption}{Assumption}
\def\BibTeX{{\rm B\kern-.05em{\sc i\kern-.025em b}\kern-.08em
    T\kern-.1667em\lower.7ex\hbox{E}\kern-.125emX}}
\begin{document}
\title{Age-Based Device Selection and Transmit Power Optimization in Over-the-Air Federated Learning}
\author{Jingyuan Liu, \IEEEmembership{Graduate Student Member, IEEE}, Zheng Chang, \IEEEmembership{Senior Member, IEEE,} Ying-Chang Liang, \IEEEmembership{Fellow, IEEE} 
\thanks{Jingyuan Liu and Zheng Chang are with School of Computer Science and Engineering, University of Electronic Science and Technology of China, Chengdu 611731, China. Z. Chang is also with University of Jyv\"askyl\"a, P. O. Box 35, FIN-40014 Jyv\"askyl\"a, Finland. Y-C. Liang is with Center for Intelligent Networking and Communications (CINC), University of Electronic Science and Technology of China, 611731 Chengdu, China.}
		

}

\markboth{Journal of \LaTeX\ Class Files,~Vol.~14, No.~8, August~2021}%
{Shell \MakeLowercase{\textit{et al.}}: A Sample Article Using IEEEtran.cls for IEEE Journals}

\maketitle

\begin{abstract}

Recently, over-the-air federated learning (FL) has attracted significant attention for its ability to enhance communication efficiency. However, the performance of over-the-air FL is often constrained by device selection strategies and signal aggregation errors. In particular, neglecting straggler devices in FL can lead to a decline in the fairness of model updates and amplify the global model's bias toward certain devices' data, ultimately impacting the overall system performance. To address this issue, we propose a joint device selection and transmit power optimization framework that ensures the appropriate participation of straggler devices, maintains efficient training performance, and guarantees timely updates. First, we conduct a theoretical analysis to quantify the convergence upper bound of over-the-air FL under age-of-information (AoI)-based device selection. Our analysis further reveals that both the number of selected devices and the signal aggregation errors significantly influence the convergence upper bound. To minimize the expected weighted sum peak age of information, we calculate device priorities for each communication round using Lyapunov optimization and select the highest-priority devices via a greedy algorithm. Then, we formulate and solve a transmit power and normalizing factor optimization problem for selected devices to minimize the time-average mean squared error (MSE). Experimental results demonstrate that our proposed method offers two significant advantages: (1) it reduces MSE and improves model performance compared to baseline methods, and (2) it strikes a balance between fairness and training efficiency while maintaining satisfactory timeliness, ensuring stable model performance.

\end{abstract}

\begin{IEEEkeywords}
Federated learning, age-of-information, over-the-air computation, device selection, resource optimization.
\end{IEEEkeywords}

\section{Introduction}
\label{sec:Introduction}
\IEEEPARstart{W}{ith} the rapid spread of smart devices equipped with advanced sensing, computing, and communication capabilities, a vast amount of distributed data is being generated locally. This data provides numerous opportunities for machine learning-driven intelligent applications such as smart healthcare \cite{baker2023artificial}, Industrial Internet of Things (IIoT) \cite{tallat2023navigating}, and autonomous driving \cite{abkenar2022survey}. However, traditional centralized machine learning methods rely on transmitting all data to a central server for unified processing, which not only demands significant communication resources but also poses serious privacy risks due to the transmission of sensitive data \cite{liu2021machine}. To address these challenges, Federated learning (FL) has emerged as a novel distributed learning paradigm \cite{mcmahan2017communication}. FL leverages the computation resources of devices to perform local training directly on edge devices, eliminating the need to transfer raw data to a central server. This approach is effective in protecting user privacy while significantly reducing communication overhead.

Although FL mitigates privacy risks and reduces communication overhead to some extent, it still faces significant communication bottlenecks caused by the high-dimensional model updates uploaded by numerous devices to the server. Many researches have been devoted to solving this problem. For example, SCAFFOLD introduces control variates to correct the local update directions of each client, significantly reducing communication overhead while ensuring the convergence of the FL model \cite{karimireddy2020scaffold}. Techniques such as gradient sparsification and quantization, as discussed in \cite{vithana2023private,tang2022gossipfl,liu2022hierarchical,chen2022energy}, greatly reduce communication costs, while client selection optimization strategies, as proposed in \cite{xu2020client,nishio2019client}, maximize bandwidth utilization and training performance. These approaches typically adopt a "transmit-then-compute" strategy for FL, which encounters difficulties in supporting massive device access under limited radio resources. To address this challenge, a low-latency multi-access scheme called over-the-air computation (AirComp) has been proposed\cite{zhu2019broadband,yang2020federated}. AirComp leverages the signal superposition property of multiple access channels, integrating model transmission and aggregation into a single step, thereby significantly enhancing the communication efficiency of FL.

In the context of over-the-air FL, various factors influence the training process and the performance of the global model. Numerous researches have explored these issues. For instance, channel and power control is a critical research direction.\cite{zhang2021gradient} proposes a dynamic power control method based on gradient statistics, while \cite{cao2022transmission} optimizes transmit power control in over-the-air FL to reduce aggregation errors. \cite{ni2022integrating} integrates non-orthogonal multiple access (NOMA) technology into over-the-air FL, enhancing communication efficiency and system performance. Meanwhile, dynamic scheduling and signal transmission optimization have also received much attention.   \cite{shao2021federated} investigates the effect of signal misalignment in over-the-air FL, and \cite{sun2021dynamic} proposes a dynamic scheduling method to balance communication efficiency and model performance. Additionally, some works focus on learning rate optimization and algorithmic frameworks. \cite{xu2021learning} introduces a dynamic learning rate optimization scheme to adapt to wireless channel variations, and \cite{zou2022knowledge} optimizes transmitter design through knowledge-guided learning. Overall, these studies have optimized the communication efficiency and model performance of over-the-air FL from multiple dimensions.

Device selection is a key research problem in FL , which plays a crucial role in determining the efficiency and effectiveness of global model training. Existing device selection strategies can be broadly categorized into two types: fair selection and discarding stragglers. Fair selection involves the server randomly selecting a group of devices for training, waiting for all selected devices to complete their computations and upload their local models, and then performing aggregation before starting the next training round, as described in \cite{albaseer2023fair,li2021ditto,mcmahan2017communication}. While this approach ensures fairness, it suffers from inefficiencies caused by stragglers. In contrast, discarding stragglers aims to improve training efficiency by excluding devices that fail to complete their training within a predefined deadline, as described in \cite{liu2021fedcpf,wu2020accelerating,nishio2019client}. However, this method often results in weaker devices being consistently excluded, thereby compromising fairness and overall model performance. An effective device selection strategy must strike a balance between training efficiency and model fairness.

Although transmission optimization and device selection in over-the-air FL have been extensively studied in the aforementioned works, their integration poses new challenges \cite{shi2024analysis,guo2022joint,yang2020federated}. In AirComp, signal aggregation errors are influenced not only by wireless channel conditions but also by the number and characteristics of participating devices \cite{shi2024analysis}. For non-IID datasets, continuously selecting the same devices can lead to weight divergence in specific directions \cite{zhao2018federated}. Moreover, while involving more devices in training can enhance the fairness and coverage of the global model, it may exacerbate signal aggregation errors and degrade overall model performance. Therefore, a carefully designed device selection strategy is critical to mitigating signal aggregation errors while maintaining fairness and training efficiency.


Age-of-information (AoI) has been widely utilized as a metric for evaluating the timeliness of device updates in FL  \cite{wu2023joint,wang2024age,ma2023channel,dong2024age}. For each device, AoI is typically defined as the time elapsed since the device was last selected for training, making it an effective indicator of update freshness. The total AoI across the system serves two primary purposes: (1) it reflects the fairness of device selection. AoI captures the sampling frequency of each device, with high AoI values indicating that certain devices have been excluded from training for extended periods, which may reveal potential fairness issues; (2) it enhances training efficiency. Optimizing the system-wide AoI helps minimize delays caused by straggler devices, thereby improving the overall efficiency of training rounds. However, existing researches have not fully explored AoI in the context of over-the-air FL. 

To the best of our knowledge, this paper is the first to incorporate AoI into over-the-air FL device selection optimization, achieving a balance between signal aggregation error, model fairness, and training efficiency. The main contributions of this paper are summarized as follows:

\begin{itemize}
	\item We theoretically analyze the convergence properties of over-the-air FL under partial device participation, establishing the relationship between the number of participating devices and signal aggregation errors with the FL performance. Based on this analysis, we formulate optimization problems to minimize the expected weighted sum of peak AoI (EWS-PAoI) and the time-average mean squared error (MSE).
	\item To solve these optimization problems, we propose the FedAirAoI method, which first uses Lyapunov optimization to calculate the device priorities for each communication round and then applies a greedy algorithm to select the highest-priority devices. Subsequently, for the selected device set, we solve the transmit power and normalizing factor optimization problem to minimize the time-average MSE. By appropriately transforming the problem and leveraging the Karush-Kuhn-Tucker (KKT) conditions, we derive closed-form solutions.
	\item We evaluate the proposed FedAirAoI strategy on CIFAR-10 and CIFAR-100 datasets. Simulation results demonstrate that FedAirAoI achieves excellent performance in terms of fairness, training efficiency, and model accuracy.
	
\end{itemize}

The remainder of this paper is organized as follows.  Section II introduces the system model and presents the results of the FL convergence analysis. Based on these results, Section III formulates the optimization problems. Section IV describes the AoI-based device selection method in detail. Section V provides a comprehensive solution to the optimization problem of minimizing the time-average MSE. Section VI presents the simulation results, and Section VII concludes the paper.

%
%
%
%
%

\label{sec:Related Work}

\section{System Model}
\label{sec:System Model}
\begin{figure}[t]
	\centering
	\includegraphics[width=1\linewidth]{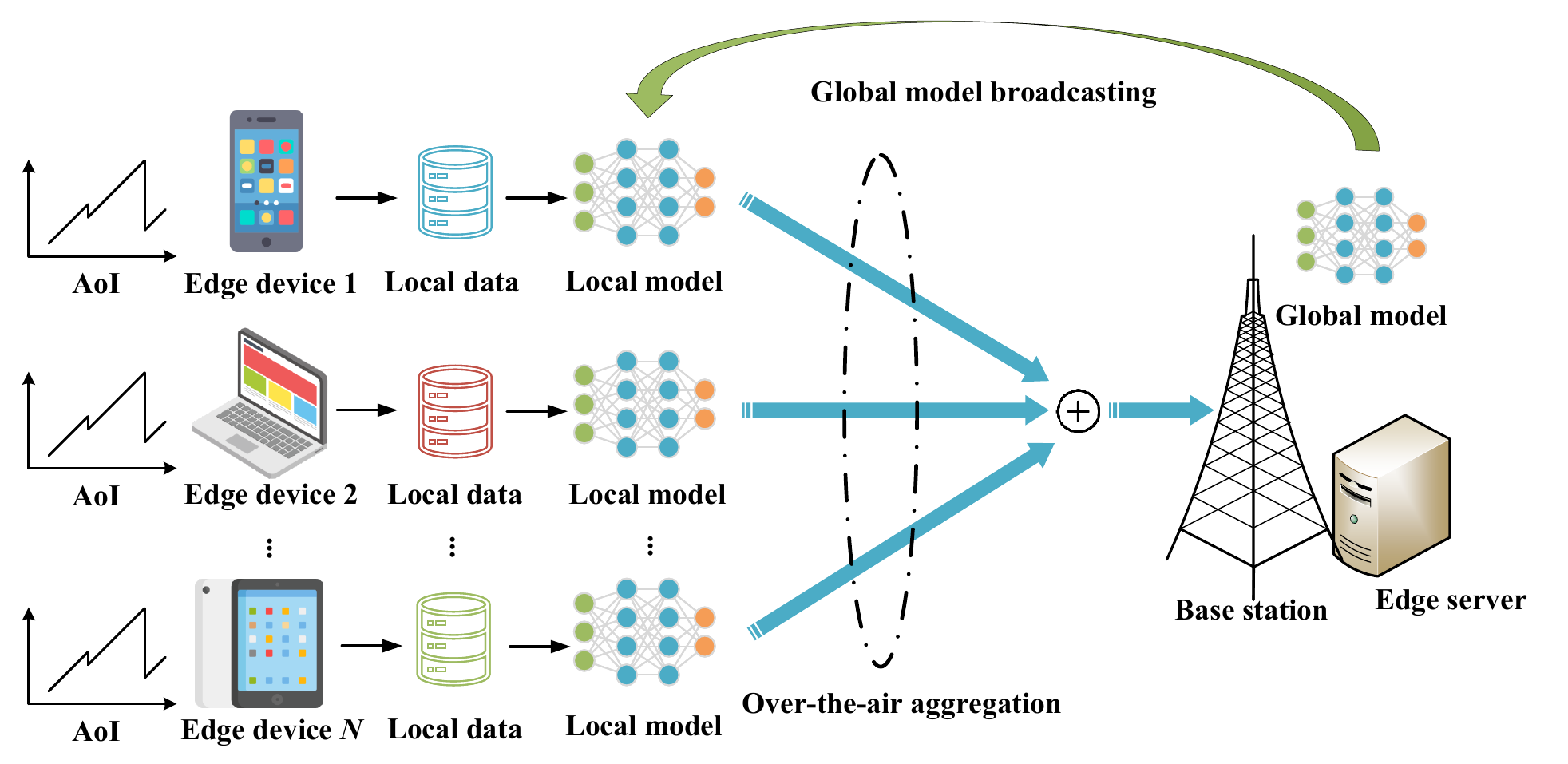}
	\caption{Over-the-air federated learning.}
	\label{fig:systemmodel}
\end{figure}
We consider an FL process in a single-cell wireless network, where an edge server at a single-antenna base station (BS) coordinates $N$ single-antenna edge devices to train a shared model. AirComp is used to aggregate model updates, enabling efficient communication, as shown in Fig. \ref{fig:systemmodel}.  In each communication round \(t\), a subset of \(K\) devices is selected from the total \(N\) devices to participate in the training process. Let \(\mathcal{N} = \{1, 2, \ldots, N\}\) represent the collection of all devices, and let \(\mathcal{S}^{(t)} \subseteq \mathcal{N}\) denote the set of \(K\) selected devices in round \(t\), where \( |\mathcal{S}^{(t)}| = K \). 

\subsection{Over-the-Air Federated Learning}

Let the global model parameter \( \boldsymbol{w}  \) represent the global weights to be optimized. Each edge device \( n \) possesses a local dataset \( \mathcal{D}_n \), consisting of \( |\mathcal{D}_n| \) data samples \( (\boldsymbol{x}_{n,i}, y_{n,i}) \), where \( \boldsymbol{x}_{n,i}  \) is the feature vector and \( y_{n,i} \) is the corresponding label. The local loss function for device \( n \) is defined as:
\begin{equation}
	F_n(\boldsymbol{w};\mathcal{D}_n)=\frac1{|\mathcal{D}_n|}\sum_{i=1}^{|\mathcal{D}_n|}f(\boldsymbol{w};(\boldsymbol{x}_{n,i},y_{n,i})).
\end{equation}

The global optimization objective of FL is defined as follows:
\begin{equation}
F(\boldsymbol{w},\mathcal{S}^{(t)}) = \sum_{n \in \mathcal{S}^{(t)}} q_n F_n(\boldsymbol{w};\mathcal{D}_n),
\end{equation}
where $q_n \ge 0$ denote the weight of the $n$-th device, such that $\sum_{n=1}^N q_n = 1$.

The local model update of device \( n \) at the \( \zeta \)-th iteration of communication round \( t \) is expressed as:
\begin{equation}
\boldsymbol{w}_n^{(t,\zeta+1)} = \boldsymbol{w}_n^{(t,\zeta)} - \lambda \tilde{\boldsymbol{g}}_n^{(t,\zeta)}, \quad \zeta = 0, \ldots, \phi - 1,
\end{equation}
where \( \lambda > 0 \) denotes the learning rate, controlling the step size of the gradient update, \( \tilde{\boldsymbol{g}}_n^{(t,\zeta)} \) is the gradient of the local loss function on device \( n \) during the \( \zeta \)-th iteration of round \( t \), defined as:
\begin{equation}
\begin{aligned}
\tilde{\boldsymbol{g}}_n^{(t,\zeta)}& =\nabla F_n(\boldsymbol{w}_n^{(t,\zeta)};\mathcal{B}_n^{(t,\zeta)}) \\
&=\frac{1}{|\mathcal{B}_n^{(t,\zeta)}|}\sum_{(\boldsymbol{x}_{\boldsymbol{n}},y_{n})\in\mathcal{B}_{n}^{(t,\zeta)}}\nabla f(\boldsymbol{w}_{n}^{(t,\zeta)};(\boldsymbol{x}_{\boldsymbol{n}},y_{n})),
\end{aligned}
\end{equation}
where \( \mathcal{B}_n^{(t,\zeta)} \subseteq \mathcal{D}_n \) is a mini-batch sampled from the local dataset \( \mathcal{D}_n \) of device \( n \),  and \( \nabla  \) represents the gradient operator, which computes the partial derivatives of a function with respect to its input variables.

The cumulative gradient update of device \( n \) in round \( t \) is given by a $d$-dimensional vector $\boldsymbol{\theta}_n^{(t)}$: 
\begin{equation}
	\boldsymbol{\theta}_n^{(t)} \in \mathbb{R}^d \triangleq \frac{\boldsymbol{w}_n^{(t, 0)} - \boldsymbol{w}_n^{(t, \phi)}}{\lambda} = \sum_{\zeta=0}^{\phi-1} \tilde{\boldsymbol{g}}_n^{(t, \zeta)}.
\end{equation}

The mean and variance of \( \boldsymbol{\theta}_n^{(t)} \) are defined as follows:
\begin{equation}
	\bar{\theta}_n^{(t)} = \frac{1}{d} \sum_{j=1}^{d} \theta_{n,j}^{(t)}.
\end{equation}
\begin{equation}
	(\pi_n^{(t)})^2 = \frac{1}{d} \sum_{j=1}^{d} \left(\theta_{n,j}^{(t)} - \bar{\theta}_n^{(t)}\right)^2.
\end{equation}

The weighted average of the cumulative gradient updates of all devices in round \( t \) is then calculated using the normalized weights:
\begin{equation}
	\bar{\theta}^{(t)} = \frac{N}{K}\sum_{n \in \mathcal{S}^{(t)}} {q}_n \bar{\theta}_n^{(t)},
\end{equation}

Similarly, the weighted variance of the cumulative gradient updates is defined as:
\begin{equation}
(\pi^{(t)})^2 = \frac{N}{K} \sum_{n \in \mathcal{S}^{(t)}} {q}_n (\pi_n^{(t)})^2,
\end{equation}

Finally, the normalized form of the cumulative gradient update \( \boldsymbol{\theta}_n^{(t)} \) for device \( n \) is given by:
\begin{equation}
\label{eq:zkt}
	\boldsymbol{z}_n^{(t)} = \frac{\boldsymbol{\theta}_n^{(t)} - \bar{\theta}^{(t)}}{\pi^{(t)}}.
\end{equation}

We consider block-fading channel model here, where the channel coefficients remain invariant within each communication round but vary independently from one round to another. In the \( t \)-th communication round, we denote the  channel gain of device $n$ as \( h_{n} ^{(t)} \). 
The estimated channel gain is given by:
\begin{equation}
|{h}_n^{(t)}|^2 = |v_n|^2 \varsigma (r_n^{(t)})^{-\vartheta} \ell^{-2}.
\end{equation}
Here, \(v_n\) is the small-scale fading coefficient, \(\varsigma\) is the frequency correlation coefficient, \(r_n^{(t)}\) represents the distance between device \(n\) and the server, \(\vartheta\) is the path loss exponent, and \(\ell^{-2}\) accounts for noise effects. 

To compensate for phase distortion due to channel fading, the transmitted signal \( \boldsymbol{z}_n^{(t)} \) is multiplied by a pre-processing factor \( \varpi_{n}^{(t)} \), defined as follows:
\begin{equation}
	\varpi_{n}^{(t)} = \sqrt{\alpha_{n}^{(t)}P_n^{\mathrm{max}}} \frac{(h_{n}^{(t)})^H}{|h_{n}^{(t)}|},
\end{equation}

\noindent where $\alpha_{n}^{(t)}\in[0,1]$ is the transmit power allocation coefficient  and \( P_n^{\mathrm{max}} \geq 0 \) denotes the maximum transmit power of device \( n \) in the \( t \)-th round. \( (h_{n}^{(t)})^H \) is the conjugate transpose of \( h_{n}^{(t)} \). This pre-processing ensures that the phase distortion introduced by the channel is countered, facilitating coherent signal combination at the receiver.

The received signal at the server, which is of dimension \( d \), can then be expressed as:
\begin{equation}
	\boldsymbol{y}^{(t)} = \sum_{n \in \mathcal{S}^{(t)}}  h_n^{(t)} \varpi_{n}^{(t)} \boldsymbol{z}_n^{(t)} + \boldsymbol{n}^{(t)},
\end{equation}

\noindent where \( \boldsymbol{n}^{(t)} \sim \mathcal{N}(\sigma^2 \boldsymbol{I}_d) \) denotes the additive white Gaussian noise (AWGN) vector, with \( \sigma^2 \) representing the noise variance. The AWGN is assumed to be independent and identically distributed (i.i.d.) across all elements of the noise vector. 

Upon receiving the aggregated signal \( \boldsymbol{y}^{(t)} \) at the server, normalization is applied to mitigate the effects of signal attenuation or scaling that occur during wireless transmission. The server uses a normalizing factor \( \eta^{(t)} \) to process the received signal, obtaining an estimate of the aggregated input signal \( \boldsymbol{z}^{(t)} \). The normalized signal is given by:
\begin{equation}
\hat{\boldsymbol{z}}^{(t)} = \frac{\boldsymbol{y}^{(t)}}{\sqrt{\eta^{(t)}}},
\end{equation}

\noindent where \( \eta^{(t)} > 0 \) is the normalization factor chosen to balance noise and signal misalignment errors. Substituting the expression for \( \boldsymbol{y}^{(t)} \), the normalized signal can be expanded as:
\begin{equation}
\hat{\boldsymbol{z}}^{(t)} = \sum_{n \in \mathcal{S}^{(t)}} \frac{\sqrt{\alpha_{n}^{(t)}P_n^{\mathrm{max}}} h_n^{(t)}}{\sqrt{\eta^{(t)}}} \boldsymbol{z}_n^{(t)} + \frac{\boldsymbol{n}^{(t)}}{\sqrt{\eta^{(t)}}}.
\end{equation}

To estimate the global gradient \( \boldsymbol{\theta}^{(t)} \), we first denormalize the received signal \( \hat{\boldsymbol{z}}^{(t)} \). The global gradient is defined as \( \boldsymbol{\theta}^{(t)} = \frac{N}{K} \sum_{n \in \mathcal{S}^{(t)}} q_n \boldsymbol{\theta}_{n}^{(t)} \), where \( q_1 = \cdots = q_N = \frac{1}{N} \) under the assumption of uniform weighting. Combined with (\ref{eq:zkt}), the estimated global gradient can be expressed as
\begin{equation}
\label{eq:egg}
\hat{\boldsymbol{\theta}}^{(t)} = \frac{1}{K} \pi^{(t)} \hat{\boldsymbol{z}}^{(t)} + \bar{\theta}^{(t)}.
\end{equation}
	
\noindent \eqref{eq:egg} can be further transformed into the following form:
	\begin{equation}
		\hat{\boldsymbol{\theta}}^{(t)} = \frac{1}{K} \pi^{(t)} (\hat{\boldsymbol{z}}^{(t)} - \boldsymbol{z}^{(t)}) + \boldsymbol{\theta}^{(t)}.
	\end{equation}

	The global model is then updated using the estimated global gradient. The update rule is given by
	\begin{equation}
	\boldsymbol{w}^{(t+1)}=\boldsymbol{w}^{(t)}-\lambda\hat{\boldsymbol{\theta}}^{(t)}=\boldsymbol{w}^{(t)}-\lambda(\boldsymbol{\theta}^{(t)}+\boldsymbol{e}^{(t)}),
	\end{equation}
\noindent where \( \boldsymbol{e}^{(t)} = \frac{1}{K} \pi^{(t)} (\hat{\boldsymbol{z}}^{(t)} - \boldsymbol{z}^{(t)}) \) represents the aggregation error in each communication round. 

\subsection{Computation and Communication Model}
At the local training stage, the computation time required by device \( n \) in round \( t \) is given by
\begin{equation}
\label{eq:comp}
	(T_{n}^\mathrm{cp})^{(t)}=\frac{\mu|\mathcal{D}_n|}{\tau_{n}^{(t)}C_n},
\end{equation}
where $\mu$ is the required central processing unit (CPU) cycles to train one sample, $\tau_{n}^{(t)}\in[0,1]$ is the computation resource allocation coefficient of device $n$ in round $t$, and $C_n$ is the available CPU cycles at device $n$, as described in \cite{wang2024age}. 

In over-the-air FL, the communication time is determined solely by the size of the local model \( D \) and the available bandwidth \( B \), as analog signals are used for transmission. This relationship can be expressed as 
\begin{equation}
\label{eq:com}
	T_{n}^\mathrm{cm}=\frac D{B}.
\end{equation}

For any device \( n \), the total time consumption in round \( t \) is given by:
\begin{equation}
\label{eq:t}
	T_{n}^{(t)}=(T_{n}^\mathrm{cp})^{(t)}+T_{n}^\mathrm{cm}.
\end{equation}
The completion time for round \( t \), \( T_c^{(t)} \), is determined by the device with the longest total time consumption among all the devices, and is expressed as:
\begin{equation}
T_c^{(t)}=\max_{n\in\mathcal{N}}\left\{\psi_{n}^{(t)}T_{n}^{(t)}\right\},
\end{equation}
where \( \psi_{n}^{(t)} \in \{0, 1\} \) is the device selection indicator. Specifically, \( \psi_{n}^{(t)} = 1 \) indicates that device \( n \) is selected in round \( t \), while \( \psi_{n}^{(t)} = 0 \) means the device is not selected.  


\subsection{Age of Information Model}
%
%
\begin{figure}[t]
	\centering
	\includegraphics[width=1\linewidth]{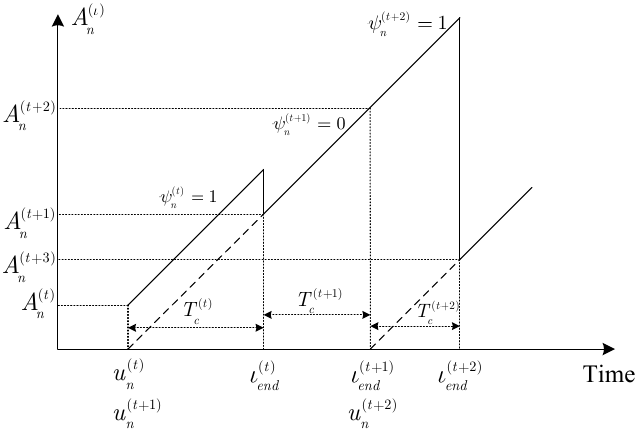}
	\caption{Variation of AoI and PAoI of device $n$ over time.}
	\label{fig:aoi}
\end{figure}

To evaluate the freshness of updates from each device in an FL system, we use the concept of the AoI. For any device \( n \) at a given elapsed time \( \iota \) within round \( t \), the AoI, denoted as \( A_n^{(\iota)} \), is defined as the elapsed time since the device was last selected to transmit an update. Thus, the AoI \( A_n^{(\iota)} \) for device \( n \) at time \( \iota \) within round \( t \) is given by:
\begin{equation}
	A_n^{(\iota)} = \iota - u_n^{(t)},
\end{equation}
where \( u_n^{(t)} \) represents the time of the last successful update for device \( n \), as shown in Fig. \ref{fig:aoi}.

At the end of round \( t \), the elapsed time is denoted by \( \iota_{\text{end}}^{(t)} \), and the AoI for device \( n \) at this point is expressed as:
\begin{equation}
	A_n^{(t+1)} = \iota_{\text{end}}^{(t)} - u_n^{(t)}.
\end{equation}

When the server receives an update from device \( n \) at the end of round \( t \), the AoI for the device is reset based on the duration of the round, \( T_c^{(t)} \), which represents the total time required to complete round \( t \). The updated PAoI \( A_n^{(t+1)} \) at the beginning of the next round is defined as:
\begin{equation}
	A_n^{(t+1)} = 
	\begin{cases} 
	A_n^{(t)} + T_c^{(t)}, & \text{if } \psi_n^{(t)} = 0, \\
	T_c^{(t)}, & \text{if } \psi_n^{(t)} = 1.
	\end{cases}
\end{equation}

For convenience, this expression can be reformulated in a recursive form as:
\begin{equation}
\label{eq:AoIt}
	A_n^{(t+1)} = (1 - \psi_n^{(t)}) A_n^{(t)} + T_c^{(t)}.
\end{equation}

This formulation effectively models the dynamics of AoI in the over-the-air FL system, capturing the impact of device selection on the freshness of information throughout the system.

\subsection{FL Convergence Analysis}
To facilitate the convergence analysis, we make the following definitions and assumptions on the loss function and gradients.
\begin{definition}
	\label{def:1}
	For device $n$ with local gradients $\{\nabla F_n(\boldsymbol{w})\}$, we define $\sigma_h^2$ to characterize the heterogeneity level of the local gradients as follows:
	\begin{equation}
	\|\nabla F_n(\boldsymbol{w})-\nabla F(\boldsymbol{w})\|^2\leq\sigma_h^2.
	\end{equation}
\end{definition}
\begin{assumption}
	\label{ass:1}
	For any parameter $\boldsymbol{w}$, the global loss function is lower bounded, i.e., $F(\boldsymbol{w})\geq$ $F(\boldsymbol{w}^*)>-\infty.$
\end{assumption}
\begin{assumption}
	\label{ass:2}
	The local loss function $F_n(\boldsymbol{w})$ is smooth with a non-negative
	constant $L$ and continuously differentiable, i.e.,
	\begin{equation}
	\|\nabla F_n(\boldsymbol{w})-\nabla F_n(\boldsymbol{w}')\|\leq L\|\boldsymbol{w}-\boldsymbol{w}'\|,\quad\forall\:\boldsymbol{w},\boldsymbol{w}'. \label{eq:smooth}
	\end{equation}
	
	Inequality (\ref{eq:smooth}) directly leads to the following inequality
    {\small
    \begin{equation}
	\label{eq:lsmooth}
	F_n(\boldsymbol{w}^{\prime})\leq F_n(\boldsymbol{w})+\langle\nabla F_n(\boldsymbol{w}),\boldsymbol{w}^{\prime}-\boldsymbol{w}\rangle+\frac L2\|\boldsymbol{w}-\boldsymbol{w}^{\prime}\|^2,\forall \boldsymbol{w},\boldsymbol{w}^{\prime}.
	\end{equation}}
	
\end{assumption}
\begin{assumption}
	\label{ass:3}
	The local mini-batch stochastic gradients $\{\tilde{\boldsymbol{g}}_n\}$ are assumed to be independent and unbiased estimates of the batch gradient $\{\nabla F_k(\boldsymbol{w}_n)\}$ with bounded variance, i.e.,
	\begin{equation}
	\begin{aligned}\mathbb{E}[\tilde{\boldsymbol{g}}_{n}]&=\nabla F_k(\boldsymbol{w}_n),\\\operatorname{Var}(\tilde{\boldsymbol{g}}_n)&=\mathbb{E}[\|\tilde{\boldsymbol{g}}_n-\nabla F_n(\boldsymbol{w}_n)\|^2]\leq\xi^2,\end{aligned}
	\end{equation}
	where $\xi\geq0$ is a constant introduced to quantify the sampling
	noise of stochastic gradients.
\end{assumption}
\begin{assumption}
	\label{ass:4}
	The variance of $d$ elements of $\boldsymbol{\theta}_n$ is upper bounded by a constant $\Gamma\geq0$, i.e., $\pi_n^2\leq\Gamma.$ 
\end{assumption}

\begin{assumption}
	\label{ass:41}
	The expected squared norm of the gradients is bounded by a constant \( G^2 \geq 0 \), i.e.,
	\begin{equation}
		\mathbb{E}\left[ \left\| \nabla F_n(\boldsymbol{w}^{(t)}) \right\|^2 \right] \leq G^2.
	\end{equation}

\end{assumption}

We consider cases where partial devices are involved. For the sake of the subsequent derivation, we assume that there is a virtual variable $\boldsymbol{v}$ that represents the global parameters in the case where all devices are involved.

\begin{assumption}
	\label{ass:5}
	The expected difference between $\boldsymbol{w}^{(t+1)}$ and $\boldsymbol{v}^{(t+1)}$ is bounded by \cite{Li2020On}
	\begin{equation}
	\mathbb{E}_{\mathcal{S}^{(t)}}\left\|\boldsymbol{w}^{(t+1)}-\boldsymbol{v}^{(t+1)}\right\|^2\leq\frac{N-K}{N-1}\frac4K\lambda \phi^2\sigma_h^2.
	\end{equation}
\end{assumption}

\begin{theorem}
	\label{theo:1}
	With Assumptions 1-6 and Definition \ref{def:1}, if learning rate $\lambda$ satisfes $\lambda\leq\min\{\frac1{2L\phi},\frac1{\sqrt{6L^2\phi^3}},\frac{\sqrt{\phi-1}}{4L\phi}\}$, then the time-average norm of the gradients after $T$ communication rounds is upper bounded by
	{\small
		\begin{equation}
		\label{eq:gap}
		\begin{aligned}
		\frac{1}{T}\sum_{t=0}^{T-1}\|\nabla F(\boldsymbol{w}^{(t)}&)  \|^2 \leq \underbrace{\frac{4(F(\boldsymbol{w}^{(0)})-F(\boldsymbol{w}^*))}{[\lambda(\phi-1)-4]T}}_{\text{Initial optimality gap}} \\
		& + \underbrace{\frac{4\phi(4\lambda^3L^2\phi^2+3\lambda)}{3[\lambda(\phi-1)-4]}\xi^2}_{\text{Variance of stochastic gradient}} \\
		& + \underbrace{\frac{16\lambda\phi^2 \sigma_h^2}{\lambda(\phi-1)-4} \left( \frac{(2 + L) (N - K)}{(N - 1) K} + \lambda^2 L^2 \phi \right)}_{\text{Gradient divergence and partial device participation}} \\
		& + \underbrace{\frac{2\lambda(1+2\lambda L)}{\lambda(\phi-1)-4} \frac{d\Gamma(K+1)}{K^2} \frac{1}{T}\sum_{t=0}^{T-1} \underbrace{\mathrm{MSE}^{(t)}}_{\text{Instantaneous MSE}}}_{\text{Time-average MSE}} \\
		& + \underbrace{\frac{4L^2 \lambda^2 }{\lambda(\phi-1)-4}G^2}_{\text{Impact of gradient norm upper bound}},
		\end{aligned}
		\end{equation}
	}
	\noindent where 
	\begin{equation}
	\label{eq:MSEt}
	\mathrm{MSE}^{(t)}=\sum_{n \in \mathcal{S}^{(t)}}  \left(\frac{ \sqrt{\alpha_{n}^{(t)}P_n^{\mathrm{max}}} h_{n}^{(t)}}{\sqrt{\eta^{(t)}}}-1\right)^2+\frac{\sigma^2}{\eta^{(t)}}.
	\end{equation}	
\end{theorem}
\begin{proof}
	See Appendix A for details.
\end{proof}

\noindent\textit{Remark.} The previous derivations in this section were based on the assumption that the sampling probabilities \( q_n \) are equal, i.e., \( q_{1}=\cdots=q_{N} = \frac{1}{N} \) for all devices. However, such an assumption may not reflect the practical selection process of devices in FL. Therefore, based on the results in \cite{Li2020On}, we extend our conclusions to scenarios where \( q_n \) may not be equal. Specifically, Theorem \ref{theo:1} still holds if the parameters \( L, \xi, \Gamma \) are replaced by  \( \tilde{L} = \nu L \), \( \tilde{\xi} = \sqrt{\nu} \xi \), and \( \tilde{\Gamma} = \sqrt{\nu} \Gamma \), where \( \nu = N \cdot \max_n q_n \).

\section{Problem Formulation}
\label{sec:Problem Formulation}
In this section, we formulate the optimization problem based on the FL convergence upper bound derived from Theorem \ref{theo:1}. The analysis of this convergence bound highlights two primary factors that impact model performance. First, the upper bound is affected by the number of devices \( K \) participating in FL training, indicating the necessity of a careful selection of the participating device set. Second, the convergence bound is influenced by the time-average MSE introduced in the analog gradient transmission, implying that minimizing the time-average MSE is essential to improve the convergence behavior.

To address the first factor, we introduce the EWS-PAoI as the optimization metric for device selection. This metric provides a way to quantify the overall freshness of the updates from each device over \( T \) training rounds. Specifically, by minimizing EWS-PAoI, we aim to ensure that the devices participating in each round are those whose updates are both timely and beneficial to the model's training process. This approach allows us to select devices strategically based on their data's relevance and timeliness, which in turn enhances model convergence.

\begin{definition}
	EWS-PAoI is defined as the weighted sum of all devices' PAoI over \( T \)-round training, expressed as:
	\begin{equation}
	\label{eq:EWS-PAoI}
	\mathbb{E} \{J\} = \mathbb{E} \left\{ \frac{1}{T N} \sum_{t=1}^T \sum_{n=1}^N q_n A_n^{(t)} \bigg| \boldsymbol{q}, \mathbf{A}^{(1)} \right\}, 
	\end{equation}
	where \( \mathbf{A}^{(t)} = [A_1^{(t)}, A_2^{(t)}, \dots, A_N^{(t)}] \) denotes the PAoI of all $N$ devices in round $t$ (at the beginning of round $t$) and \( \boldsymbol{q} = [q_1, q_2, \dots, q_N] \)  denotes the weight vector of all devices, satisfying \( q_n \in [0,1] \) and \( \sum_{n=1}^N q_n = 1 \).  The weights reflect the relative importance of each device based on the characteristics of their datasets.
\end{definition}

Thus, the problem of PAoI-based device selection for FL can be formulated as an optimization problem \textbf{P1} as follows, where \( \psi_n^{(t)} \) represents the decision variable:
\begin{align}
	\textbf{P1: } & \min\limits_{\psi_n^{(t)}} \mathbb{E} \{ J \} \label{p1}\\
	\text{s.t. } & \psi_n^{(t)} \in \{0, 1\}, \quad \forall n \in \mathcal{N}, \ \forall t, \tag{\ref{p1}{a}}\\
	& A_n^{(t+1)} = (1 - \psi_n^{(t)}) A_n^{(t)} + T_c^{(t)}, \quad \forall n \in \mathcal{N}, \ \forall t \tag{\ref{p1}{b}}.
\end{align}

Next, we discuss the second factor, the time-average MSE due to the analog gradient transmission. The MSE contributes directly to the convergence bound and impacts the accuracy of aggregated gradients. By minimizing this error, we can tighten the convergence upper bound (\ref{eq:gap}), resulting in better training performance. To formalize this objective, we rewrite the time-average MSE (\ref{eq:MSEt}) as:
\begin{equation}
\label{eq:avgMSE}
	\overline{\mathrm{MSE}}=\frac{1}{T}\sum_{t=0}^{T-1}\left[\sum_{n \in \mathcal{S}^{(t)}}  \left(\frac{ \sqrt{\alpha_{n}^{(t)}P_n^{\mathrm{max}}} h_{n}^{(t)}}{\sqrt{\eta^{(t)}}}-1\right)^2+\frac{\sigma^2}{\eta^{(t)}}\right].
\end{equation}

Specifically, our objective is to jointly optimize the transmit power allocated to each device and the receive normalizing factor at the server side, ensuring that the gradients received by the server are as accurate as possible while avoiding energy waste of the devices. The optimization problem \textbf{P2} can be formulated as follows:
\begin{align}
\textbf{P2: } &\min_{\alpha_{n}^{(t)}, \eta^{(t)}} \    \overline{\mathrm{MSE}} \label{p2}\\
\text{s.t. } & \alpha_{n}^{(t)}\in[0,1],\  \forall n\in\mathcal{N}, \forall t, \tag{\ref{p2}{a}} \\
& 0\leq\frac1T\sum_{t=0}^{T-1}\alpha_{n}^{(t)}P_n^{\mathrm{max}}\leq\bar{P}_n,\ \forall n\in\mathcal{N},  \tag{\ref{p2}{b}}\\
& \eta^{(t)}\geq0,\ \forall t \tag{\ref{p2}{c}},
\end{align}

Based on the observations of \textbf{P1} and \textbf{P2}, there is no coupling of variables between these two problems, allowing them to be optimized independently. We begin by solving \textbf{P1} to determine the selected device set. Next, we perform the optimization for \textbf{P2} based on the selected device set, obtaining the transmit power allocation and the receive normalizing factor. It is worth noting that \textbf{P1} and \textbf{P2} are solved sequentially, with \textbf{P2} taking the device set determined by \textbf{P1} as its input.

\section{Age of Information Based Device Selection}
The optimization problem \textbf{P1} is formulated as a combinatorial optimization problem (COP) for device selection with discrete decision variables. Given the NP-hard nature of this problem, the search space expands exponentially as the number of devices \( N \) and training rounds \( T \) increase. Moreover, there is an inherent coupling between the device selection decisions and the completion time of each FL training round, introducing additional complexity to solving \textbf{P1} using traditional COP methods. Inspired by the analytical approaches of \cite{dong2024age} and \cite{kadota2018scheduling}, the original problem can be transformed to optimize the lower bound of \(\lim_{T \to \infty} \mathbb{E}\{J\}\). This, in turn, allows the problem to be decomposed into two sub-problems: calculating the priority of each device and determining the optimal number of devices to select based on their priority scores. Adopting this approach enables a more efficient solution for device selection in each FL training round.
\subsection{Determination of Device Priority in Each Round}
Lyapunov optimization can be used to solve this sub-problem \cite{kadota2018scheduling}. To apply it to our model, we first define some relevant terms. The objective is to obtain a value of \( \mathbb{E}\left\{J\right\} \) as low as possible, which is referred to as the desired system state. The device selection in each round is called a decision-making stage. The current PAoI values before making the selection decision are regarded as the current system state. The purpose of Lyapunov optimization is to push the system towards the desired state at each decision-making stage. In order to naturally quantify the degree of deviation from the system state, a quadratic Lyapunov function is defined as follows:
\begin{equation}
	L\left(\mathbf{A}^{(t)}\right)=\frac1N\sum_{n=1}^Nq_n\Big(A_n^{(t)}\Big)^2.
\end{equation}

The one-frame Lyapunov drift is:
\begin{equation}
	\Delta\left(\mathbf{A}^{(t)}\right)=\mathbb{E}\left\{L\left(\mathbf{A}^{(t+1)}\right)-L\left(\mathbf{A}^{(t)}\right)\right\}.
\end{equation}

The Lyapunov drift represents the change in the system state from the current stage to the next stage, indicating how the expected system state evolves. Minimizing the drift reduces the degree of deviation from the desired system state in each round, enabling the system to gradually approach the desired state. Based on this theory, we use the concept of device priority to measure each device’s potential contribution to optimizing the system state. High-priority devices are considered more beneficial in the current round for reducing the PAoI, so selecting these devices preferentially can more effectively minimize the Lyapunov drift, thereby guiding the system toward the desired low-PAoI state.

\begin{definition}
	\label{def:3}
	The priority of device $n$ at round $t$ is defined as follows:
	\begin{equation}
		\varphi_n^{(t)} = \frac{q_n A_n^{(t)}}{T_{n}^{(t)}},
	\end{equation}
	where $q_n$ and $A_n^{(t)}$ denote the weight and PAoI value of the device $n$, respectively.  $T_{n}^{(t)}$ denotes the time consumed by device $n$ in round $t$. We assume that the server can obtain the computation resource allocation coefficient $\tau_{n}^{(t)}$ of device $n$ in round $t$. Then, the total time consumption $T_{n}^{(t)}$ of device $n$ in round $t$ can be obtained by  (\ref{eq:comp}), (\ref{eq:com}) and (\ref{eq:t}).
\end{definition}

\noindent\textit{Remark.} The priority setting defined in Definition \ref{def:3} helps to avoid the selection of devices with high PAoI but limited computation power. Since the total time consumption \(T_n^{(t)} \) of such devices may be higher, their priority \(\varphi_n^{(t)} \) is consequently lower. Thus, such a priority setting balances the freshness of information  and computation efficiency of the system.

\subsection{Determination of the Number of Participating Devices in Each Round}
The previous subsection concluded that selecting high-priority devices can minimize the Lyapunov drift in the current round, thereby pushing the system toward the desired low-PAoI state. This subsection thus focuses on determining the optimal number of high-priority devices to select in the current round, with the goal of minimizing the weighted sum PAoI (WS-PAoI) in round \( t+1 \). The analysis will start by examining the form of WS-PAoI in round \( t+1 \). Based on (\ref{eq:EWS-PAoI}) and (\ref{eq:AoIt}), the WS-PAoI for round $t+1$ is as follows:
\begin{equation}
\label{eq:WS-PAoI}
	\begin{aligned}
	r^{(t+1)}& =\frac{1}{N}\sum_{n=1}^{N}q_{n}A_n^{(t+1)} \\
	&=\frac{1}{N}\sum_{n=1}^{N}q_{n}T_{c}^{(t)}+\frac{1}{N}\sum_{n=1}^{N}q_{n}(1 - \psi_n^{(t)}) A_n^{(t)} \\
	&=\frac1NT_c^{(t)}+\frac1N\sum_{n\notin\mathcal{S}^{(t)}}q_n A_n^{(t)}.
	\end{aligned}
\end{equation}

Let $\mathcal{U}^{(t)}$ denote the set of the top \( K^{(t)} \) highest-priority devices. Since these devices are prioritized for selection, (\ref{eq:WS-PAoI}) can be rewritten as:
\begin{equation}
\label{eq:ut}
	r^{(t+1)}=\frac1NT_c^{(t)}+\frac1N\sum_{n\notin\mathcal{U}^{(t)}}q_n A_n^{(t)}.
\end{equation}

The objective of the problem is to find the optimal \( K^{(t)} \) (denoted by \( K_{\text{opt}}^{(t)} \)) such that the WS-PAoI is minimized, i.e.
\begin{equation}
	K_{\text{opt}}^{(t)}=\underset{K_{\text{opt}}^{(t)}}{\operatorname*{\arg\min}}\left\{r^{(t+1)}\right\}.
\end{equation}

Then, we analyze the two terms in (\ref{eq:ut}). Regarding the first term, it is determined by the longest completion time of the devices in $\mathcal{U}^{(t)}$. We can define a function $T_{c}^{(t)}=\Psi\left(K^{(t)}|\varphi_n^{(t)}\right)$ to represent this relationship. For the second term, it can be divided into two parts based on the completion time of devices not in $\mathcal{U}^{(t)}$. Thus,   (\ref{eq:ut}) can be transformed into:
\begin{equation}
\label{eq:vw}
	r^{(t+1)}=\frac1NT_c^{(t)}+\frac1N\sum_{n\in\mathcal{V}^{(t)}}q_nA_n^{(t)}+\frac1N\sum_{n\in\mathcal{W}^{(t)}}q_nA_n^{(t)},
\end{equation}
where $\mathcal{V}^{(t)}$ represents those not in \( \mathcal{U}^{(t)} \) but with completion time not exceeding \( \Psi\left(K^{(t)}|\varphi_n^{(t)}\right) \). In contrast, $\mathcal{W}^{(t)}$ represents those not in \( \mathcal{U}^{(t)} \) and with completion time greater than \( \Psi\left(K^{(t)}|\varphi_n^{(t)}\right) \). Their formal definitions are as follows:
\begin{equation}
	\mathcal{V}^{(t)}=\left\{n\left|T_{n}^{(t)}\leq\Psi\left(K^{(t)}|\varphi_n^{(t)}\right),n\notin\mathcal{U}^{(t)},n\in[1,N]\right.\right\},
\end{equation}
\begin{equation}
\label{eq:W}
\mathcal{W}^{(t)}=\left\{n\left|T_{n}^{(t)}>\Psi\left(K^{(t)}|\varphi_n^{(t)}\right),n\notin\mathcal{U}^{(t)},n\in[1,N]\right.\right\}.
\end{equation}

We observe that including the devices in \( \mathcal{V}^{(t)} \) as part of the selected device set does not impact the first term in equation (\ref{eq:vw}). Specifically, the completion time of the devices in \( \mathcal{V}^{(t)} \) will not exceed \( \Psi\left(K^{(t)}|\varphi_n^{(t)}\right)\), which leaves \( T_c^{(k)} \) unchanged. By including \( \mathcal{V}^{(t)} \) as part of the selected device set, the second term is effectively eliminated, resulting in the following simplified form: 
\begin{equation}
\label{eq:simp}
r^{(t+1)}=\frac1NT_c^{(t)}+\frac1N\sum_{n\in\mathcal{W}^{(t)}}q_nA_n^{(t)},
\end{equation}

To minimize (\ref{eq:simp}), an intuitive approach is to calculate \( r^{(t+1)} \) for all values of \( K^{(t)} \in [1, N] \) and identify the optimal \( K^{(t)} \). To further reduce the search space of the algorithm, we will examine the property of the function \( T_{c}^{(t)} = \Psi\left(K^{(t)} | \varphi_n^{(t)}\right) \), which can help in excluding some non-optimal solutions.

\begin{algorithm}[t]
	\caption{Greedy Priority-Aware Device Selection}
	\begin{algorithmic}[1]
		\State \textbf{Input:} \( \varphi_n^{(t)}, T_n^{(t)} \), \( n \in [1, N] \)
		\State \textbf{Output:} \( K_{\text{opt}}^{(t)} \) and \( \mathcal{S}^{(t)} \)
		\State Initialize \( \mathcal{Q} = \emptyset \), \( \mathcal{U}^{(t)} = \emptyset \), \( \mathcal{W}^{(t)} = \emptyset \), \( k = 1 \), \( t_{\text{max}} = 0 \)
		\State Sort devices based on \( \varphi_n^{(t)} \) in descending order
		\State Append \( 1^\prime \) to \( \mathcal{U}^{(t)} \)
		\State Set \( t_{\text{max}} =  T_{1^\prime}^{(t)} \)
		\For{\( n = 2 \) to \( N \)}
		\If{ \( t_{\text{max}} <  T_{n^\prime}^{(t)} \) }
		\State Build \( \mathcal{W}^{(t)} \) based on (\ref{eq:W})
		\State Calculate \( r \) based on (\ref{eq:WS-PAoI}), where \( T_c^{(t)} = t_{\text{max}} \), \(\mathcal{S}_n^{(t)} = \mathcal{U}^{(t)} \cup \mathcal{W}^{(t)} \)
		\State Append the triplet \( \{k : \mathcal{S}_n^{(t)} : r\} \) to \( \mathcal{Q} \)
		\State \( t_{\text{max}} =  T_{n^\prime}^{(t)} \)
		\EndIf
		\State Append \( n^\prime \) to \( \mathcal{U}^{(t)} \)
		\State \( k = k + 1 \)
		\EndFor
		\State Append the triplet \( \{N : \mathcal{S}_N^{(t)} : t_{\text{max}}/N \} \) to \( \mathcal{Q} \)
		\State \( K_{\text{opt}}^{(t)} \) and \( \mathcal{S}^{(t)} \) are the first and the second component of the element in \( \mathcal{Q} \) with the minimal third component, respectively
		\State \textbf{return} \( K_{\text{opt}}^{(t)} \) and \( \mathcal{S}^{(t)} \)
	\end{algorithmic}
	\label{alg:1}
\end{algorithm}

\noindent\textbf{Property 1.} \( \Psi\left(K^{(t)} | \varphi_n^{(t)}\right) \) is a non-decreasing staircase function with respect to \( K^{(t)} \).

\begin{proof}
	Reorder the devices based on their priority values \( \varphi_n^{(t)} \), and let \( n^\prime \) denote the ordered indices such that \( \varphi_{1^\prime}^{(t)} \geq \varphi_{2^\prime}^{(t)} \geq \dots \geq \varphi_{N^\prime}^{(t)} \). According to the definition of \( T_{c}^{(t)} \), which is based on the maximum completion time of the selected devices in round \( t \), we have \( \Psi\left(K^{(t)} | \varphi_n^{(t)}\right) = \max \left\{ T_{1^\prime}^{(t)}, T_{2^\prime}^{(t)}, \dots, T_{{K^{(t)}}^\prime}^{(t)} \right\} \). When \( K^{(t)} \) increases to \( K^{(t)} + 1 \), the function \( \Psi\left(K^{(t)} + 1 | \varphi_n^{(t)}\right) \) becomes \( \Psi\left(K^{(t)} + 1 | \varphi_n^{(t)}\right) = \max \left\{ T_{1^\prime}^{(t)}, T_{2^\prime}^{(t)}, \dots, T_{{K^{(t)}}^\prime}^{(t)}, T_{{(K^{(t)}+1)}^\prime}^{(t)} \right\} \). Since adding another device to the set can only increase or leave unchanged the maximum value of the completion times, it follows that \( \Psi\left(K^{(t)} + 1 | \varphi_n^{(t)}\right) \geq \Psi\left(K^{(t)} | \varphi_n^{(t)}\right) \). Therefore, \( \Psi\left(K^{(t)} | \varphi_n^{(t)}\right) \) is non-decreasing with respect to \( K^{(t)} \), which completes the proof.
\end{proof}

Using Property 1, if \( T_{{(K^{(t)}+1)}^\prime}^{(t)} > \Psi\left(K^{(t)} | \varphi_n^{(t)}\right) \) is not satisfied, (\ref{eq:simp}) decreases as \( K^{(t)} \) increases, since the first term remains constant while the second term decreases. However, if \( T_{{(K^{(t)}+1)}^\prime}^{(t)} > \Psi\left(K^{(t)} | \varphi_n^{(t)}\right) \) holds, (\ref{eq:simp}) may increase due to the rise in the first term. Consequently, values of \( K^{(t)} \) that satisfy \( T_{{(K^{(t)}+1)}^\prime}^{(t)} > \Psi\left(K^{(t)} | \varphi_n^{(t)}\right) \)  represent possible candidates for minimizing  (\ref{eq:simp}) and should be examined. 

In summary, we propose Algorithm \ref{alg:1} to determine the optimal number of devices and the set of selected devices by prioritizing high-priority devices. This algorithm uses a greedy approach to iteratively select devices, aiming to minimize the WS-PAoI. The details are presented in Algorithm \ref{alg:1}.

\section{MSE Minimization via Transmit Power and Normalizing Factor Optimization}

After determining the selected device set \( \mathcal{S}^{(t)} \) and the optimal number of devices \( K_{\text{opt}}^{(t)} \) in each communication round through Algorithm \ref{alg:1}, the focus of this section is shifted to optimizing the MSE for the given device set. The challenge in solving \textbf{P2} arises from the trade-offs and complexity of the optimization process. Specifically, the objective function (\ref{eq:avgMSE}) combines noise-induced error and signal misalignment error, which requires careful balancing between these two components. While increasing the receive normalizing factor \(\eta^{(t)} \) helps reduce the noise-induced error, it also amplifies the signal misalignment error. In addition, aligning the signal amplitude requires adjusting the transmit power of the  device, and the transmit power allocation coefficient is limited by the maximum power (\ref{p2}a) and the average power (\ref{p2}b), so perfect alignment cannot be achieved in many cases. In addition, due to the strong coupling between the transmit power allocation coefficients and the receive normalizing factors, the problem is non-convex, adding great complexity to the optimization process. Therefore, a reasonable transformation and decomposition of the original problem is required.
\subsection{Optimization of Receive Normalizing factor}
First, we keep the transmit power allocation coefficient $\alpha_{n}^{(t)}$ of the device fixed and optimize the receive normalizing factor $\eta^{(t)}$. The optimization problem is as follows:
{\small
\begin{equation}
	\textbf{P2.1}: \min_{\{\eta^{(t)}\geq0\}}\sum_{t=0}^{T-1}\left[\sum_{n=1}^{K_{\text{opt}}^{(t)}}\left(\frac{\sqrt{\alpha_{n}^{(t)}P_n^{\mathrm{max}}}|h_n^{(t)}|}{\sqrt{\eta^{(t)}}}-1\right)^2+\frac{\sigma^2}{\eta^{(t)}}\right].
\end{equation}}

Since the optimization variable \(\eta^{(t)}\) appears independently in each communication round \(t\), the problem \textbf{P2.1} can be naturally decomposed into \(T\) independent subproblems, one for each communication round. Each subproblem is expressed as:
\begin{equation}
	\min_{\eta^{(t)}\geq0} \mathcal{E}(\eta^{(t)}) \triangleq \sum_{n=1}^{K_{\text{opt}}^{(t)}}\left(\frac{\sqrt{\alpha_{n}^{(t)}P_n^{\mathrm{max}}}|h_n^{(t)}|}{\sqrt{\eta^{(t)}}}-1\right)^2+\frac{\sigma^2}{\eta^{(t)}}.
\end{equation}
By introducing the substitution \(\Omega^{(t)} = 1/\sqrt{\eta^{(t)}}\), the objective function is rewritten as:
\begin{equation}
\label{eq:omega}
	\mathcal{E}(\Omega^{(t)}) = \sum_{n=1}^{K_{\text{opt}}^{(t)}} \left(\sqrt{\alpha_{n}^{(t)}P_n^{\mathrm{max}}}|h_n^{(t)}|\Omega^{(t)}- 1\right)^2 + \left(\sigma\Omega^{(t)}\right)^2.
\end{equation}
(\ref{eq:omega}) is a convex objective function with respect to \(\Omega^{(t)}\). To determine the optimal value of \(\Omega^{(t)*}\), the first-order derivative of \(\mathcal{E}(\Omega^{(t)})\) is set to zero. Consequently, a closed-form solution for the optimal \(\Omega^{(t)*}\) is derived, and the corresponding optimal receive normalizing factor \(\eta^{(t)*}\) is expressed as:
\begin{equation}
\label{eq:eta}
	\eta^{(t)*} = \frac{1}{(\Omega^{(t)*})^2} = \left( \frac{\sigma^2 + \sum_{n=1}^{K_{\text{opt}}^{(t)}} (\sqrt{\alpha_{n}^{(t)}P_n^{\mathrm{max}}}|h_n^{(t)}|)^2}{\sum_{n=1}^{K_{\text{opt}}^{(t)}} \sqrt{\alpha_{n}^{(t)}P_n^{\mathrm{max}}}|h_n^{(t)}|} \right)^2.
\end{equation}
\subsection{Optimization of Transmit Power Allocation Coefficients}
We fix the receive normalizing factor \(\eta^{(t)} \) and focus on optimizing the transmit power allocation coefficients \( \alpha_{n}^{(t)} \) for each device by addressing the following optimization problem:
\begin{align}
\textbf{P2.2: } &\min_{\alpha_{n}^{(t)}}  \sum_{t=0}^{T-1} \left[ \sum_{n=1}^{K_{\text{opt}}^{(t)}} \left( \frac{\sqrt{\alpha_{n}^{(t)}P_n^{\mathrm{max}}}|h_n^{(t)}|}{\sqrt{\eta^{(t)}}} - 1 \right)^2 \right] \label{p22}\\
\text{s.t. } & \text{(\ref{p2}{a}), (\ref{p2}{b}).}\nonumber
\end{align}

It can be observed that the objective function (\ref{p22}) is separable across devices. Specifically, the contribution of each device's transmit power allocation coefficients \( \alpha_{n}^{(t)} \) to the overall objective function is independent. Therefore, the original problem \textbf{P2.2} can be decomposed into $K_{\text{opt}}^{(t)}$ smaller subproblems, each of which corresponds to a specific device \( n \):
\begin{align}
\min_{\alpha_{n}^{(t)}} & \sum_{t=0}^{T-1} \left( \frac{\sqrt{\alpha_{n}^{(t)}P_n^{\mathrm{max}}}|h_n^{(t)}|}{\sqrt{\eta^{(t)}}} - 1 \right)^2 \label{p23}\\\text{s.t. } & 0\leq\alpha_{n}^{(t)}\leq 1,\   \forall t, \tag{\ref{p23}{a}}\\&0\leq\sum_{t=0}^{T-1}\alpha_{n}^{(t)}P_n^{\mathrm{max}}\leq T \bar{P}_n. \tag{\ref{p23}{b}}
\end{align}

Each subproblem aims to optimize the transmit power allocation coefficients \( \alpha_{n}^{(t)} \) of a single device while satisfying its associated constraints. Since the decomposed subproblems are independent, the optimization for each device can be performed separately. Furthermore, these subproblems are convex, making them well-suited for solution using the KKT conditions. 

The Lagrangian function for the optimization problem (\ref{p23}) is given by:
{\small
\begin{equation}
\begin{aligned}
	&\mathcal{L}(\{\alpha_{n}^{(t)}\}, \gamma_n, \{\rho_{n}^{(t)}\}) = \sum_{t=0}^{T-1} \left( \frac{\sqrt{\alpha_{n}^{(t)} P_n^{\mathrm{max}}} \, |h_n^{(t)}|}{\sqrt{\eta^{(t)}}} - 1 \right)^2 \\
	&\quad+ \sum_{t=0}^{T-1} \rho_{n}^{(t)} (\alpha_{n}^{(t)} - 1)+\gamma_n \left( \sum_{t=0}^{T-1} \alpha_{n}^{(t)} P_n^{\mathrm{max}} - T\bar{P}_n \right) ,
\end{aligned}
\end{equation}}
where \( \rho_{n}^{(t)} \geq 0 \) and \( \gamma_n \geq 0 \) are the Lagrange multipliers. Then \( \{\alpha_{n}^{(t)*}\}, \{\rho_{n}^{(t)*}\}, \gamma_n^*   \) should satisfy the following KKT conditions:
\begin{align}
& 0 \leq \alpha_{n}^{(t)*} \leq 1,  \label{KKT1} \\
& 0 \leq \sum_{t=0}^{T-1} \alpha_{n}^{(t)*}  \leq \frac{T\bar{P}_n}{P_n^{\mathrm{max}}} , \label{KKT2} \\
& \rho_{n}^{(t)*} \geq 0, \gamma_n^* \geq 0  ,  \label{KKT3} \\
& \gamma_n^* \left( \sum_{t=0}^{T-1} \alpha_{n}^{(t)*}  - \frac{T\bar{P}_n}{P_n^{\mathrm{max}}} \right) = 0, \label{KKT4} \\
& \rho_{n}^{(t)*} \left( \alpha_{n}^{(t)*} - 1 \right) = 0, \label{KKT5} \\
& \left( \dfrac{\sqrt{\alpha_{n}^{(t)*}P_n^{\mathrm{max}}} \, |h_n^{(t)}|}{\sqrt{\eta^{(t)}}} - 1 \right) \dfrac{P_n^{\mathrm{max}} |h_n^{(t)}|}{\sqrt{\eta^{(t)} \alpha_{n}^{(t)*} P_n^{\mathrm{max}}}} + \gamma_n^* P_n^{\mathrm{max}} \nonumber\\\quad&+ \rho_{n}^{(t)*} = 0. \label{KKT6}
\end{align}

\begin{algorithm}[t]
	\caption{Iterative Algorithm for MSE Optimization}
	\begin{algorithmic}[1]
		\State \textbf{Input:} \(\{h_n^{(t)}\}_{t=0}^{T-1}\), stopping condition \(\epsilon_0\).
		\State \textbf{Output:} \(\{\eta^{(t)}\}\) and \(\{\alpha_{n}^{(t)}\}\).
		\State Initialize: Transmit power allocation coefficients \(\{\alpha_{n}^{(t)}\}^0\) and \(i = 0\).
		\Repeat
		\State \(i = i + 1\).
		\State Given \(\{\alpha_{n}^{(t)}\}^{i-1}\), update \(\{\eta^{(t)}\}^i\).
		\State Given \(\{\eta^{(t)}\}^i\), update \(\{\alpha_{n}^{(t)}\}^i\). 
		\Until \( \frac{\overline{\mathrm{MSE}}^{i-1} - \overline{\mathrm{MSE}}^i}{\overline{\mathrm{MSE}}^i} < \epsilon_0 \).
	\end{algorithmic}
\label{alg:2}
\end{algorithm}

By analyzing the KKT conditions, we can obtain the optimal solution of (\ref{p23}).
\begin{theorem}
	The optimal solution of $\alpha_{n}^{(t)*}$ is:
	\begin{equation}
	\label{theo:2}
	\left.\alpha_{n}^{(t)*}=\begin{cases}\min\left\{\frac{\eta^{(t)}}{P_n^{\mathrm{max}}|h_n^{(t)}|^2},1\right\},\\[2ex]\text{if }\sum_{t=0}^{T-1}\min\left\{\frac{\eta^{(t)}}{P_n^{\mathrm{max}}|h_n^{(t)}|^2},1\right\}< \frac{T\bar{P}_n}{P_n^{\mathrm{max}}},\\[2ex]\min\left\{\left(\frac{\sqrt{\eta^{(t)}}|h_n^{(t)}|}{\sqrt{P_n^{\mathrm{max}}}\left(|h_n^{(t)}|^2+\gamma_n^*\eta^{(t)}\right)}\right)^2,1\right\},\\[2ex]\text{if }\sum_{t=0}^{T-1}\min\left\{\frac{\eta^{(t)}}{P_n^{\mathrm{max}}|h_n^{(t)}|^2},1\right\}\geq\frac{T\bar{P}_n}{P_n^{\mathrm{max}}},\end{cases}\right.
	\end{equation}
	where $\gamma_n^*$ can be found to ensure the average power constraint $\sum_{t=0}^{T-1}\alpha_{n}^{(t)*}= \frac{T\bar{P}_n}{P_n^{\mathrm{max}}}$ via the bisection search method.
\end{theorem}
\begin{proof}
	See Appendix C for details.
\end{proof}

The overall description of the solution is given by Algorithm \ref{alg:2}. Algorithm \ref{alg:2} aims to minimize the MSE by iteratively optimizing the parameters. The algorithm alternates between two steps: updating the receive normalizing factors and refining the power allocation coefficients. The process is iteratively repeated until convergence, as determined by a predefined stopping criterion based on the relative change in MSE across iterations.

\noindent\textit{Remark.} To further enhance the practicality of our solution, we introduce an online optimization method for the transmit power allocation coefficients \( \alpha_{n}^{(t)} \), making them depend only on the current round's channel state information (CSI) \( h_n^{(t)} \). This is achieved by dynamically updating the Lagrange multipliers \( \gamma_n^{(t)} \) within each communication round. Instead of solving for a fixed \( \gamma_n^* \) that requires knowledge of the CSI over all rounds, we update \( \gamma_n^{(t)} \) based on the instantaneous power usage and the average power constraint in an online fashion. Specifically, after obtaining \( \alpha_{n}^{(t)} \) from the per-round optimization, we adjust \( \gamma_n^{(t)} \) using a simple iterative rule that accounts for deviations from the average power budget. This online optimization approach decouples the optimization across time slots, allowing \( \alpha_{n}^{(t)} \) to be computed using only the current CSI and the updated \( \gamma_n^{(t)} \). By dynamically adjusting the Lagrange multipliers online, we ensure that the average power constraint is satisfied over time while maintaining the integrity of our original solution. This modification results in an efficient algorithm suitable for real-time implementation, as it leverages our existing derivations and requires minimal additional computation.

\section{Performance Evaluation}
\subsection{Simulation Setting}
In the simulations, the wireless channels between the edge devices and the edge server across different communication rounds are modeled as independent and identically distributed (i.i.d.) Rayleigh fading channels. The total number of edge devices is set to \( N = 20 \). The signal-to-noise ratio (SNR) for each edge device \( n \) is defined as \( \text{SNR}_n = \frac{\bar{P}_n}{\sigma^2} \), with a default value of 10 dB \cite{cao2020optimized}. The maximum transmit power for each device is set to \( P_n^{\text{max}} = 3\bar{P}_n \). The stopping condition \(\epsilon_0\) is set to $10^{-5}$. The required CPU cycles $\mu$ to train one sample is set to $10^7$. The available CPU cycles $C_n$ at device $n$ is set to $0.5$ GHz. The available bandwidth $B$ is set to $20$ MHz. The size of the local model $D$ is set to $11.7$ M.

We utilize the non-IID CIFAR-10 and CIFAR-100 datasets to train the global model based on the ResNet-18 architecture. ResNet-18 is a classic convolutional neural network comprising 18 layers, including convolutional layers, batch normalization layers, and fully connected layers, specifically designed for image classification tasks. The weight \( q_n \) is set based on \cite{deng2021auction}, where we apply exponential weighting according to the number of data classes \( M_n \) contained in device \( n \). The formula is given as  $q_n = \frac{2^{M_n}}{\sum_{n} 2^{M_n}}$.

We conduct a comparative analysis with four baselines, detailed as follows. The baselines differ from our proposed method only in the specific aspects highlighted below, while all other settings remain consistent with the proposed FedAirAoI strategy.
\begin{itemize}
	\item Full Power: Each device transmits with a fixed power \(\bar{P}_n\), and the receive normalizing factor \(\eta^{(t)}\) is computed by substituting \(\bar{P}_n\) into  (\ref{eq:eta}).
	\item Channel Inversion \cite{li2019wirelessly}: Channel inversion aims to compensate for the channel gain of each device, eliminating the interference caused by the channel to the transmitted signal and ensuring consistent signal amplitude at the receiver. Specifically, when the channel gain is poor (e.g., \( |h_n^{(t)}|^2 \) is very small), the device requires higher transmit power to compensate for the channel loss. If the channel condition of the device is too poor, such that even using the maximum transmit power \( P_n^{\max} \) cannot meet the requirement, the device is deactivated.
	\item FedAvg \cite{mcmahan2017communication}: As a fair device selection method, FedAvg randomly selects a group of clients for training. The server can only perform aggregation and start the next round after all the selected devices have completed their computations and uploaded their local models.
	\item HybridFL \cite{wu2020accelerating}: As a device-dropping selection method, the concept of HybridFL is to drop selectedm devices in a round if their completion times exceed the deadline.
\end{itemize}
\subsection{Simulation Results}


\begin{figure}[t]
	\centering
	\includegraphics[width=0.8\linewidth]{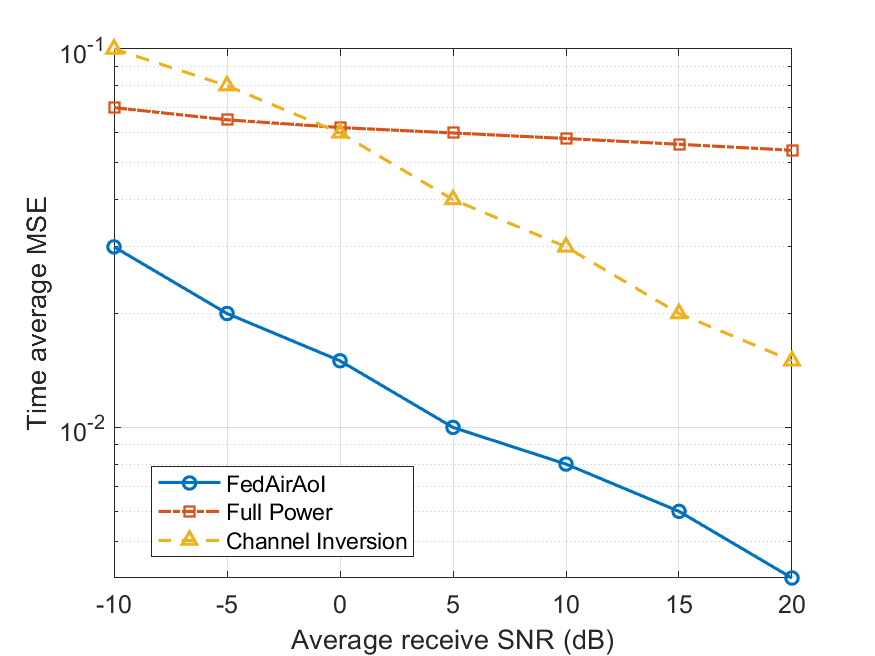}
	\caption{Time-average MSE variation with SNR.}
	\label{fig:exp1}
\end{figure}

\begin{figure}[t]
	\centering
	\includegraphics[width=0.8\linewidth]{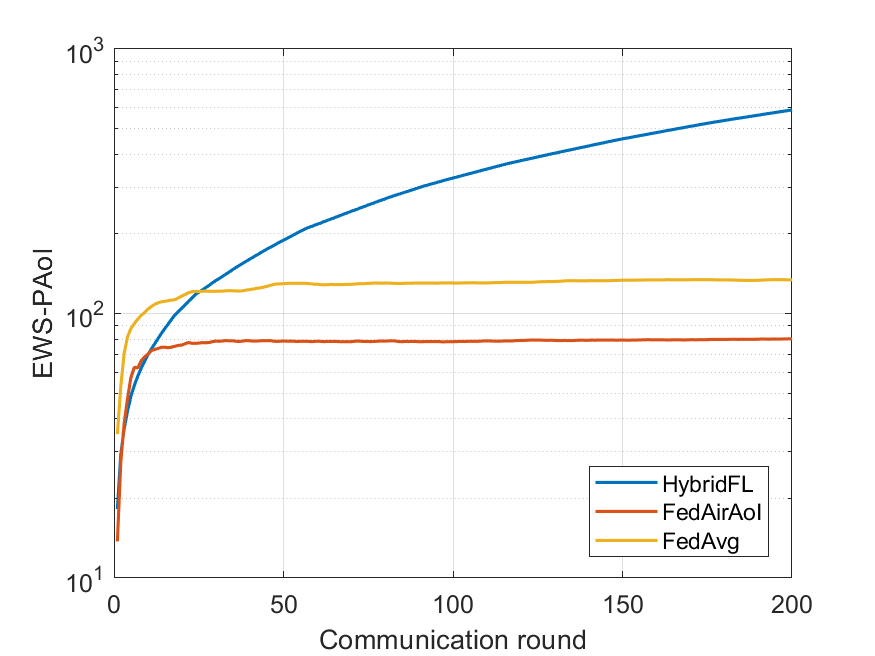}
	\caption{EWS-PAoI over communication rounds.}
	\label{fig:exp2}
\end{figure}

\begin{figure}[t]
	\centering
	\includegraphics[width=0.8\linewidth]{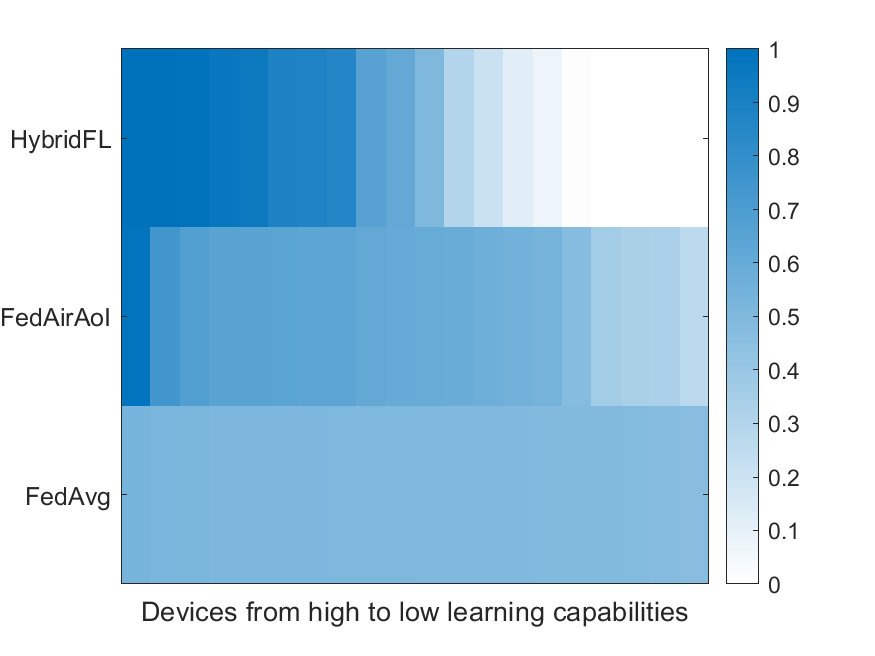}
	\caption{Device selection frequency under different methods.}
	\label{fig:exp3}
\end{figure}
\begin{figure}[t]
	\centering
	\includegraphics[width=0.8\linewidth]{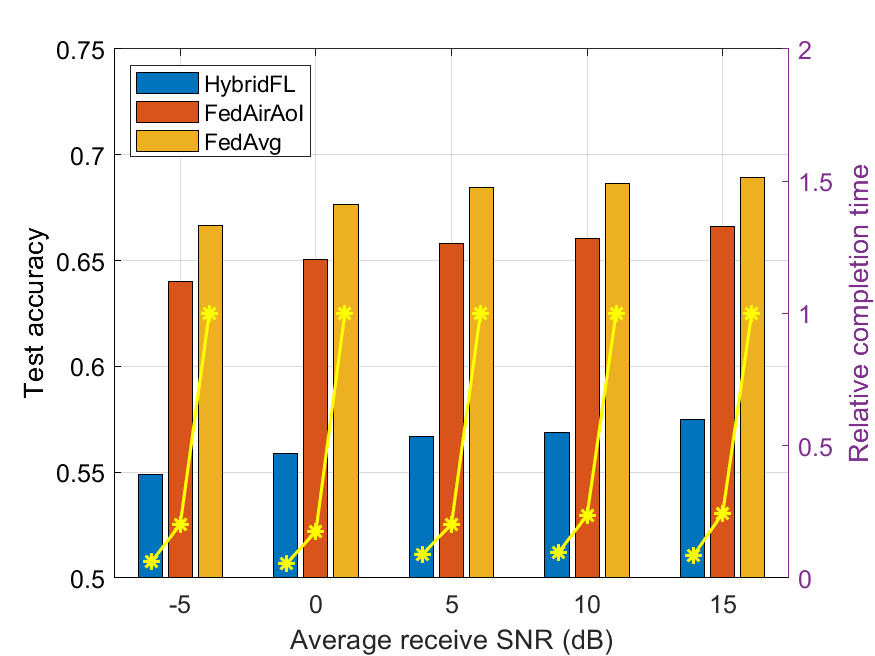}
	\caption{Impact of SNR on test accuracy and relative completion time on CIFAR-10 dataset.}
	\label{fig:exp4}
\end{figure}
In Fig. \ref{fig:exp1}, we show how the time average MSE changes with the average receive SNR. It can be observed that the time average MSE of all methods decreases as the average receive SNR increases. Across the entire SNR range, FedAirAoI outperforms the other two baseline methods. Under low SNR conditions, the Full Power method performs better than the Channel Inversion method. This is because full-power transmission significantly suppresses noise-induced errors, which dominate the MSE in such scenarios. As the average receive SNR increases, the performance gap between FedAirAoI and the Full Power method widens, indicating that the Full Power method is less adaptable to varying SNR conditions.

Fig. \ref{fig:exp2} illustrates the variation of EWS-PAoI across communication rounds for three methods. The EWS-PAoI of HybridFL increases significantly with the number of communication rounds. This is due to HybridFL’s inability to effectively handle stragglers, leading to a sustained increase in the PAoI of unselected straggler devices. In comparison, FedAvg performs slightly worse than FedAirAoI because it fails to account for the AoI in its selection policy, causing the stragglers to impact the overall update freshness. On the other hand, the EWS-PAoI of FedAirAoI stabilizes in later communication rounds, maintaining relatively low values. This is attributed to its fair device selection strategy and reduced impact of stragglers.

\begin{figure*}[t]
	\centering
	\begin{subfigure}[b]{0.45\textwidth}
		\centering
		\includegraphics[width=\textwidth]{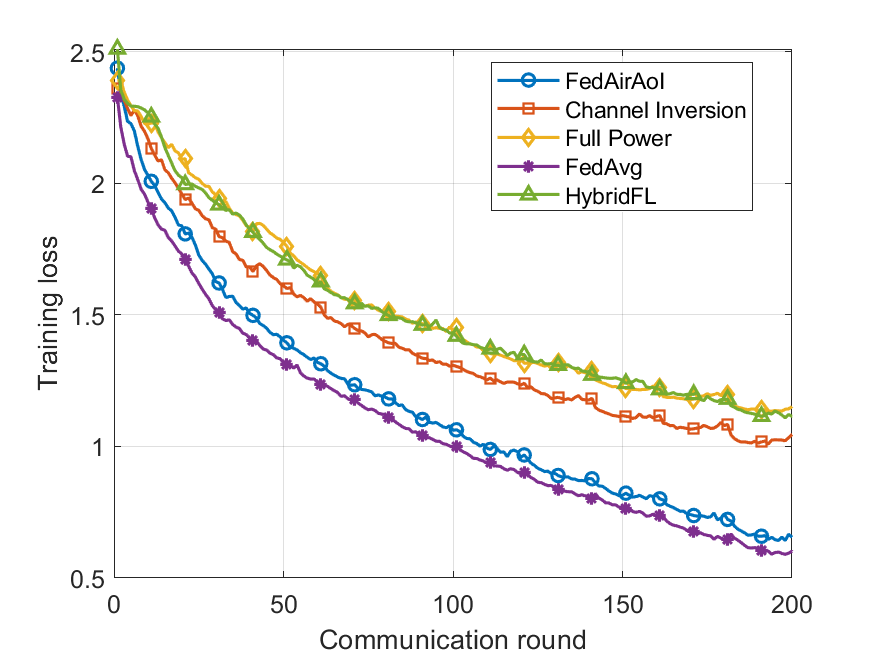}
		\caption{Training loss comparison on CIFAR-10 dataset.}
		\label{fig:exp51}
	\end{subfigure}
	\hfill
	\begin{subfigure}[b]{0.45\textwidth}
		\centering
		\includegraphics[width=\textwidth]{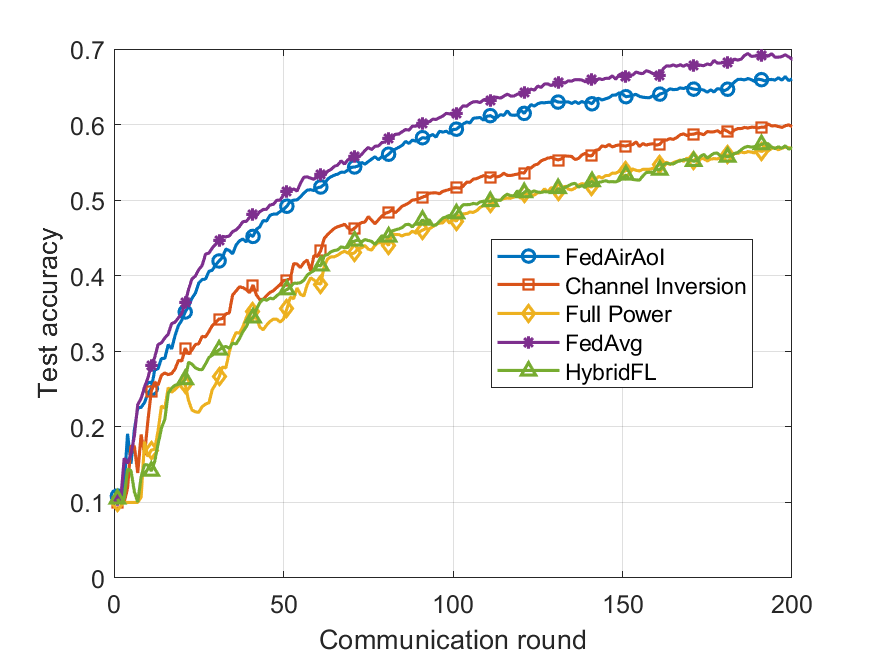}
		\caption{Test accuracy comparison on CIFAR-10 dataset.}
		\label{fig:exp52}
	\end{subfigure}
	
	\vspace{0.5em} 
	
	\begin{subfigure}[b]{0.45\textwidth}
		\centering
		\includegraphics[width=\textwidth]{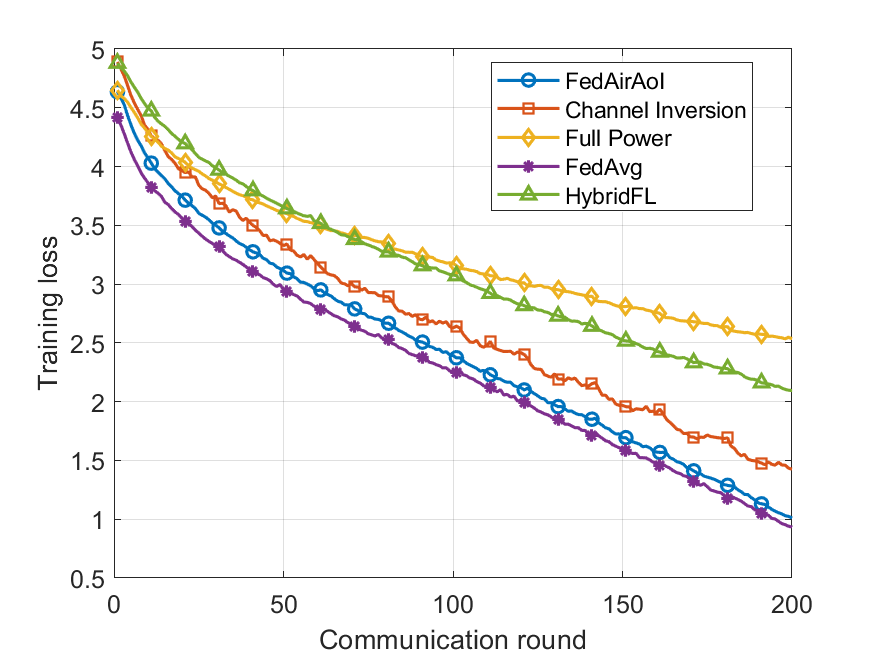}
		\caption{Training loss comparison on CIFAR-100 dataset.}
		\label{fig:exp61}
	\end{subfigure}
	\hfill
	\begin{subfigure}[b]{0.45\textwidth}
		\centering
		\includegraphics[width=\textwidth]{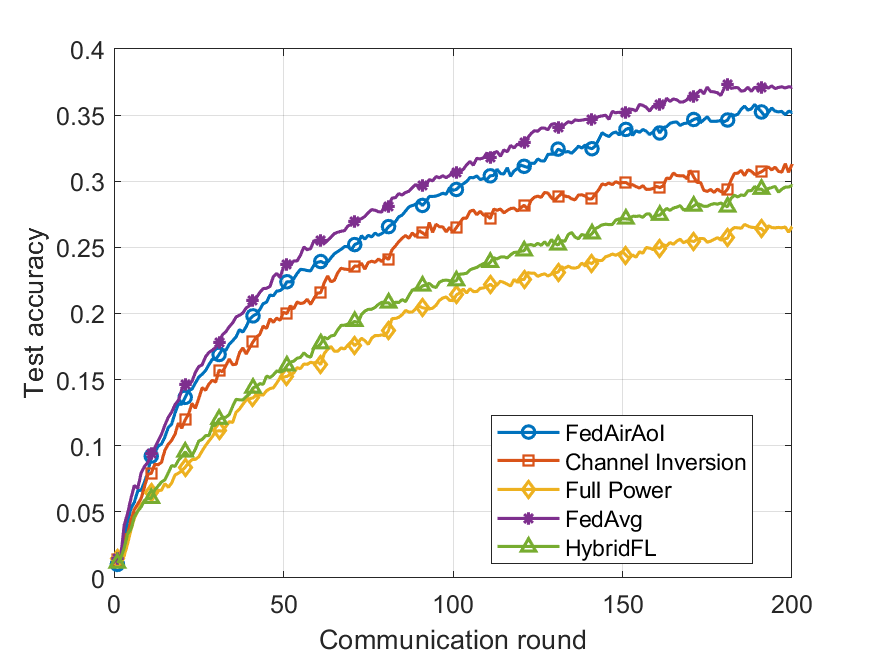}
		\caption{Test accuracy comparison on CIFAR-100 dataset.}
		\label{fig:exp62}
	\end{subfigure}
	
	\caption{Comparison of different methods on CIFAR-10 and CIFAR-100 datasets.}
	\label{fig:combined_exp_cifar}
\end{figure*}

Fig. \ref{fig:exp3} highlights the differences in device selection frequency among the three methods, with devices sorted from high to low learning capabilities. In HybridFL, the selection frequency is concentrated on devices with high learning capabilities, while devices with low learning capabilities are rarely selected. In contrast, FedAvg disperses the selection frequency more uniformly across all devices, effectively mitigating extreme bias. However, this approach may reduce overall training efficiency due to the negative impact of devices with low learning capabilities. FedAirAoI, on the other hand, employs a more balanced device selection strategy, markedly enhancing the selection frequency for devices with lower learning capabilities while concurrently preserving a relatively high selection frequency for those with higher learning capabilities.

To evaluate training efficiency, the average completion time (\( \bar{T_c} \)) of a training round is used as the performance metric, calculated as \( \bar{T_c} = \frac{1}{T} \sum_{t=1}^T T_c^{(t)} \). In Fig. \ref{fig:exp4}, we use relative completion time as the metric. For the other two methods, their relative completion time is calculated as the ratio of their respective average completion time to that of FedAvg. Fig. \ref{fig:exp4} presents the performance of the three methods in terms of test accuracy and relative completion time under varying average SNR conditions. As the SNR increases, FedAvg consistently achieves the highest test accuracy, but its relative completion time remains the highest due to its unoptimized device selection strategy. HybridFL achieves the shortest completion time, but its focus on selecting devices with high learning capabilities may lead to insufficient data diversity and reduced model generalization. FedAirAoI strikes a good balance between accuracy and completion time. For instance, when SNR = 10 dB, FedAirAoI achieves a test accuracy slightly lower than FedAvg by 2.62\% but significantly higher than HybridFL by 9.15\%. At the same time, its relative completion time is 76.3\% lower than FedAvg. This demonstrates that FedAirAoI effectively mitigates the impact of stragglers through optimized device selection while balancing data diversity and efficiency.

In Fig. \ref{fig:combined_exp_cifar}, we present the training loss and test accuracy of different methods on the CIFAR-10 and CIFAR-100 datasets. The results show that our proposed method, FedAirAoI, outperforms the Channel Inversion and Full Power methods across both datasets. This advantage stems from FedAirAoI’s optimization of the convergence upper bound for FL. Both the Channel Inversion and Full Power methods introduce significant signal aggregation errors during over-the-air aggregation, which negatively impact FL performance. While the Channel Inversion method does not directly optimize the convergence upper bound of FL, it mitigates signal aggregation errors to some extent by applying a channel threshold. In contrast, the Full Power method lacks any optimization, resulting in poorer performance. Additionally, the performance of FedAirAoI lies between that of HybridFL and FedAvg,aligning with the previous observations.

\label{sec:performance}

\section{Conclusion}
\label{sec:conclusion}
In this paper, we propose an AoI-based device selection and transmit power optimization framework for over-the-air FL to address the challenges of maintaining fairness, timeliness, and communication efficiency in resource-constrained wireless networks. By theoretically analyzing the convergence properties of over-the-air FL, we highlight the impact of device participation and signal aggregation errors on the convergence upper bound. Based on this, we developed the FedAirAoI framework, which integrates the minimization of the EWS-PAoI and a transmit power allocation strategy to reduce the time-average MSE. Simulation results demonstrate that FedAirAoI achieves a superior balance between training efficiency, fairness, and model performance compared to baseline methods. By stabilizing EWS-PAoI and reducing MSE, our approach ensures timely updates from participating devices while effectively mitigating the impact of stragglers and signal aggregation errors. 

{
	
	\section*{Appendix}
	
	\subsection{Proof of Theorem \ref{theo:1}}
	Before proving Theorem \ref{theo:1}, we fırst present the following
	three useful lemmas. Lemma \ref{lem:1}, \ref{lem:2} and \ref{lem:3} are proved in \cite{zou2022knowledge}.  Lemma  \ref{lem:4} is proved in Appendix B.
	\begin{lemma}
		\label{lem:1}
		With Assumption \ref{ass:2} and \ref{ass:3}, the following inequality holds
		\begin{equation}
		\begin{aligned}
		&\lambda\mathbb{E}[\langle\nabla F(\boldsymbol{w}^{(t)}),\boldsymbol{\theta}^{(t)}\rangle] \leq \\
		&\quad -\frac{\lambda\phi}{2}\|\nabla F(\boldsymbol{w}^{(t)})\|^2 
		- \frac{\lambda}{2\phi}\mathbb{E}[\|\boldsymbol{\theta}^{(t)}\|^2] \\
		&\quad + \frac{\lambda L^2}{K}\sum_{n \in \mathcal{S}^{(t)}}\sum_{\zeta=0}^{\phi-1}\mathbb{E}[\|\boldsymbol{w}^{(t)}-\boldsymbol{w}_n^{(t,\zeta)}\|^2] \\
		&\quad + \lambda\phi\xi^2.
		\end{aligned}
		\end{equation} 
	\end{lemma}
	\begin{lemma}
		\label{lem:2}
		With Assumption \ref{ass:2}, \ref{ass:3} and Definition \ref{def:1}, the difference between the global model vector and the individual local model vector is bounded, i.e.,
		\begin{equation}
		\begin{aligned}
		\sum_{\zeta=0}^{\phi-1}\mathbb{E}[\|\boldsymbol{w}^{(t)}-\boldsymbol{w}_n^{(t,\zeta)}\|^2] &\leq 
		\frac{4\lambda^2\xi^2\phi^3}{3} + 4\lambda^2\sigma_h^2\phi^3 \\
		&\quad + 4\lambda^2\phi^3\|\nabla F(\boldsymbol{w}^{(t)})\|^2.
		\end{aligned}
		\end{equation}
	\end{lemma}
	\begin{lemma}
		\label{lem:3}
		With Assumptions \ref{ass:3} and \ref{ass:4}, the aggregation error
		and the instantaneous MSE has the following relationship
		\begin{equation}
		\mathbb{E}[\|\boldsymbol{e}^{(t)}\|^2]\leq d\frac{\Gamma(K+1)}{K^2}\mathrm{MSE}^{(t)}.
		\end{equation}
	\end{lemma}
	\begin{lemma}
		\label{lem:4}
		With Assumptions \ref{ass:2} and \ref{ass:41}, the following inequality holds:
		\begin{equation}
		\begin{aligned}
		\langle \nabla F(\boldsymbol{v}^{(t+1)}),\ \boldsymbol{w}^{(t+1)} - \boldsymbol{v}^{(t+1)} \rangle &\leq \frac{\left\| \nabla F(\boldsymbol{w}^{(t)}) \right\|^2 + L^2 \lambda^2 G^2}{2} \\
		&+ \left\| \boldsymbol{w}^{(t+1)} - \boldsymbol{v}^{(t+1)} \right\|^2.
		\end{aligned}
		\end{equation}
		\noindent\textit{Remark.} A similar inequality can be derived for the term \(\langle \nabla F(\boldsymbol{w}^{(t+1)}), \boldsymbol{v}^{(t+1)} - \boldsymbol{w}^{(t+1)} \rangle\). Specifically, using similar steps as in the proof of Lemma \ref{lem:4}, we can obtain the following bound:
		\begin{equation}
		\begin{aligned}
		\langle \nabla F(\boldsymbol{w}^{(t+1)}), \boldsymbol{v}^{(t+1)} - \boldsymbol{w}^{(t+1)} \rangle &\leq \frac{\left\| \nabla F(\boldsymbol{w}^{(t)}) \right\|^2 + L^2 \lambda^2 G^2}{2} \\
		&\quad + \left\| \boldsymbol{v}^{(t+1)} - \boldsymbol{w}^{(t+1)} \right\|^2.
		\end{aligned}
		\end{equation}	
	\end{lemma}
	
	\noindent\textit{\textbf{Proof of Theorem \ref{theo:1}}}: $F(\boldsymbol{w})$ is $L$-smooth and we have the following inequality
	\begin{equation}
	\begin{aligned}
	&F(\boldsymbol{w}^{(t+1)}) - F(\boldsymbol{w}^{(t)}) = \left[F(\boldsymbol{w}^{(t+1)}) - F(\boldsymbol{v}^{(t+1)})\right] \\
	& + \left[F(\boldsymbol{v}^{(t+1)}) - F(\boldsymbol{w}^{(t+1)})\right] + \left[F(\boldsymbol{w}^{(t+1)}) - F(\boldsymbol{w}^{(t)})\right] \\
	&\stackrel{(a)}{\leq} \langle \nabla F(\boldsymbol{v}^{(t+1)}), \boldsymbol{w}^{(t+1)} - \boldsymbol{v}^{(t+1)} \rangle + \frac{L}{2} \|\boldsymbol{w}^{(t+1)} - \boldsymbol{v}^{(t+1)}\|^2 \\
	& + \langle \nabla F(\boldsymbol{w}^{(t+1)}), \boldsymbol{v}^{(t+1)} - \boldsymbol{w}^{(t+1)} \rangle + \frac{L}{2} \|\boldsymbol{v}^{(t+1)} - \boldsymbol{w}^{(t+1)}\|^2 \\
	& + \langle \nabla F(\boldsymbol{w}^{(t)}), \boldsymbol{w}^{(t+1)} - \boldsymbol{w}^{(t)} \rangle + \frac{L}{2} \|\boldsymbol{w}^{(t+1)} - \boldsymbol{w}^{(t)}\|^2 \\
	&\stackrel{(b)}{\leq} \langle \nabla F(\boldsymbol{v}^{(t+1)}), \boldsymbol{w}^{(t+1)} - \boldsymbol{v}^{(t+1)} \rangle \\
	&+ \langle \nabla F(\boldsymbol{w}^{(t+1)}), \boldsymbol{v}^{(t+1)} - \boldsymbol{w}^{(t+1)} \rangle \\
	& + L \|\boldsymbol{w}^{(t+1)} - \boldsymbol{v}^{(t+1)}\|^2 - \lambda \langle\nabla F(\boldsymbol{w}^{(t)}),\boldsymbol{\theta}^{(t)}\rangle \\
	& + \frac{\lambda}{2}\|\nabla F(\boldsymbol{w}^{(t)})\|^2 + \left(\frac{\lambda}{2} + \lambda^2L\right)\|\boldsymbol{e}^{(t)}\|^2 + \lambda^2L\|\boldsymbol{\theta}^{(t)}\|^2 \\
	&\stackrel{(c)}{\leq} \left( 2 + L \right) \|\boldsymbol{w}^{(t+1)} - \boldsymbol{v}^{(t+1)}\|^2 - \lambda \langle\nabla F(\boldsymbol{w}^{(t)}),\boldsymbol{\theta}^{(t)}\rangle \\
	& + \left(1+\frac{\lambda}{2}\right)\|\nabla F(\boldsymbol{w}^{(t)})\|^2 + \left( \frac{\lambda}{2} + \lambda^2L \right)\|\boldsymbol{e}^{(t)}\|^2 \\
	& + \lambda^2L\|\boldsymbol{\theta}^{(t)}\|^2 + L^2 \lambda^2 G^2.
	\end{aligned}
	\end{equation}
	where $(a)$ follows from (\ref{eq:lsmooth}), $(b)$ follows from \cite{zou2022knowledge},  $(c)$ follows from Lemma \ref{lem:4}. By taking expectations over
	stochastic sampling and receiver noise at both sides, we obtain
	\begin{equation}
	\begin{aligned}
	&\mathbb{E}[F(\boldsymbol{w}^{(t+1)}) - F(\boldsymbol{w}^{(t)})] \leq \left( 2 + L \right) \mathbb{E} \|\boldsymbol{w}^{(t+1)} - \boldsymbol{v}^{(t+1)}\|^2 \\
	&\quad - \lambda \mathbb{E}\left[\langle\nabla F(\boldsymbol{w}^{(t)}),\boldsymbol{\theta}^{(t)}\rangle\right] + \left(1 + \frac{\lambda}{2}\right) \|\nabla F(\boldsymbol{w}^{(t)})\|^2 \\
	&\quad + \left(\frac{\lambda}{2} + \lambda^2 L\right) \mathbb{E}[\|\boldsymbol{e}^{(t)}\|^2] + \lambda^2 L \mathbb{E}[\|\boldsymbol{\theta}^{(t)}\|^2] + L^2 \lambda^2 G^2.
	\end{aligned}
	\end{equation}
	
	Combined with Assumption \ref{ass:5}, Lemma \ref{lem:1}, \ref{lem:2}, \ref{lem:3} and $\lambda\leq \min\{\frac1{2L\phi},\frac1{\sqrt{6L^2\phi^3}},\frac{\sqrt{\phi-1}}{4L\phi}\}$, we obtain
	\begin{equation}
	\label{eq:single}
	\begin{aligned}
	&\mathbb{E}[F(\boldsymbol{w}^{(t+1)}) - F(\boldsymbol{w}^{(t)})] \leq \\
	&\quad 4 \sigma_h^2 \left( \frac{(2 + L) (N - K)}{(N - 1) K} \lambda \phi^2 + \lambda^3 L^2 \phi^3 \right) \\
	&\quad + \left(1 - \frac{\lambda(\phi-1)}{4}\right)\|\nabla F(\boldsymbol{w}^{(t)})\|^2 \\
	&\quad + \left(\frac{4\lambda^3L^2\phi^3}{3} + \lambda\phi\right)\xi^2 \\
	&\quad + \left(\frac{\lambda}{2} + \lambda^2L\right)d\frac{\Gamma(K+1)}{K^2}\mathrm{MSE}^{(t)} + L^2 \lambda^2 G^2.
	\end{aligned}
	\end{equation}
	
	By summing up (\ref{eq:single}) for all $T$ communication rounds and rearranging the terms, we have
	\begin{equation}
	\begin{aligned}
	&\mathbb{E}[F(\boldsymbol{w}^{(T)}) - F(\boldsymbol{w}^{(0)})] \leq \\ &\left(1 - \frac{\lambda(\phi-1)}{4}\right) \sum_{t=0}^{T-1} \|\nabla F(\boldsymbol{w}^{(t)})\|^2 \\
	&\quad + \left(\frac{4\lambda^3 L^2 \phi^3}{3} + \lambda \phi\right) \xi^2 T + L^2 \lambda^2 G^2 T \\
	&\quad + 4 \sigma_h^2 \left( \frac{(1 + L) (N - K)}{(N - 1) K} \lambda \phi^2 + \lambda^3 L^2 \phi^3 \right) T \\
	&\quad + \left(\frac{\lambda}{2} + \lambda^2 L\right) \sum_{t=0}^{T-1} d\frac{\Gamma(K+1)}{K^2} \mathrm{MSE}^{(t)}.
	\end{aligned}
	\end{equation}

	With Assumption \ref{ass:1}, we have $F(\boldsymbol{w}^{(T)}) - F(\boldsymbol{w}^{(0)}) \geq F(\boldsymbol{w}^*) - F(\boldsymbol{w}^{(0)})$, which yields the result.

	\subsection{Proof of Lemma \ref{lem:4}}
	We aim to bound the inner product \( \langle \nabla F(\boldsymbol{v}^{(t+1)}),\ \boldsymbol{w}^{(t+1)} - \boldsymbol{v}^{(t+1)} \rangle \) in terms of \( \left\| \nabla F(\boldsymbol{w}^{(t)}) \right\|^2 \) and known constants. Using the Cauchy-Schwarz inequality, we can obtain:
	\begin{equation}
	\begin{aligned}
	&\left| \langle \nabla F(\boldsymbol{v}^{(t+1)}),\ \boldsymbol{w}^{(t+1)} - \boldsymbol{v}^{(t+1)} \rangle \right| \leq\\& \left\| \nabla F(\boldsymbol{v}^{(t+1)}) \right\| \cdot \left\| \boldsymbol{w}^{(t+1)} - \boldsymbol{v}^{(t+1)} \right\|.
	\end{aligned}
	\end{equation}

	Next, we apply Young's inequality for any \( \varepsilon > 0 \):
	\begin{equation}
	ab \leq \frac{a^2}{2\varepsilon} + \frac{\varepsilon b^2}{2}.
	\end{equation}
	
	Let \( a = \left\| \nabla F(\boldsymbol{v}^{(t+1)}) \right\| \) and \( b = \left\| \boldsymbol{w}^{(t+1)} - \boldsymbol{v}^{(t+1)} \right\| \). Then:
	\begin{equation}
	\begin{aligned}
	& \left\| \nabla F(\boldsymbol{v}^{(t+1)}) \right\| \cdot \left\| \boldsymbol{w}^{(t+1)} - \boldsymbol{v}^{(t+1)} \right\| \leq  \\
	&\frac{1}{2\varepsilon} \left\| \nabla F(\boldsymbol{v}^{(t+1)}) \right\|^2 + \frac{\varepsilon}{2} \left\| \boldsymbol{w}^{(t+1)} - \boldsymbol{v}^{(t+1)} \right\|^2.
	\end{aligned}
	\end{equation}
	
	Thus,
	\begin{equation}
	\begin{aligned}
	& \left| \langle \nabla F(\boldsymbol{v}^{(t+1)}),\ \boldsymbol{w}^{(t+1)} - \boldsymbol{v}^{(t+1)} \rangle \right| \leq \\
	& \frac{1}{2\varepsilon} \left\| \nabla F(\boldsymbol{v}^{(t+1)}) \right\|^2 
	+ \frac{\varepsilon}{2} \left\| \boldsymbol{w}^{(t+1)} - \boldsymbol{v}^{(t+1)} \right\|^2.
	\end{aligned}
	\end{equation}
	
	Using Assumption \ref{ass:2},
	\begin{equation}
	\left\| \nabla F(\boldsymbol{v}^{(t+1)}) - \nabla F(\boldsymbol{w}^{(t)}) \right\| \leq L \left\| \boldsymbol{v}^{(t+1)} - \boldsymbol{w}^{(t)} \right\|.
	\end{equation}
	
	Therefore,
	\begin{equation}
	\left\| \nabla F(\boldsymbol{v}^{(t+1)}) \right\| \leq \left\| \nabla F(\boldsymbol{w}^{(t)}) \right\| + L \left\| \boldsymbol{v}^{(t+1)} - \boldsymbol{w}^{(t)} \right\|.
	\end{equation}
	
	Squaring both sides,
	\begin{equation}
	\left\| \nabla F(\boldsymbol{v}^{(t+1)}) \right\|^2 \leq \left( \left\| \nabla F(\boldsymbol{w}^{(t)}) \right\| + L \left\| \boldsymbol{v}^{(t+1)} - \boldsymbol{w}^{(t)} \right\| \right)^2.
	\end{equation}
	
	Expanding the right-hand side:
	\begin{equation}
	\begin{aligned}
	\left\| \nabla F(\boldsymbol{v}^{(t+1)}) \right\|^2 &\leq \left\| \nabla F(\boldsymbol{w}^{(t)}) \right\|^2 \\
	&\quad + 2 L \left\| \nabla F(\boldsymbol{w}^{(t)}) \right\| \left\| \boldsymbol{v}^{(t+1)} - \boldsymbol{w}^{(t)} \right\| \\
	&\quad + L^2 \left\| \boldsymbol{v}^{(t+1)} - \boldsymbol{w}^{(t)} \right\|^2.
	\end{aligned}
	\end{equation}
	
	To eliminate the mixed term, we apply Young's inequality again with \( \delta > 0 \):
	\begin{equation}
	2ab \leq \frac{a^2}{\delta} + \delta b^2.
	\end{equation}
	
	Let \( a = \left\| \nabla F(\boldsymbol{w}^{(t)}) \right\| \), \( b = L \left\| \boldsymbol{v}^{(t+1)} - \boldsymbol{w}^{(t)} \right\| \), and choose \( \delta = 1 \):
	\begin{equation}
	\begin{aligned}
	2 L \left\| \nabla F(\boldsymbol{w}^{(t)}) \right\| \left\| \boldsymbol{v}^{(t+1)} - \boldsymbol{w}^{(t)} \right\| &\leq \left\| \nabla F(\boldsymbol{w}^{(t)}) \right\|^2 \\
	&\quad + L^2 \left\| \boldsymbol{v}^{(t+1)} - \boldsymbol{w}^{(t)} \right\|^2.
	\end{aligned}
	\end{equation}
	
	Substitute back:
	\begin{equation}
	\left\| \nabla F(\boldsymbol{v}^{(t+1)}) \right\|^2 \leq \alpha \left\| \nabla F(\boldsymbol{w}^{(t)}) \right\|^2 + \beta,
	\end{equation}
	where $\alpha=2$, $\beta=2 L^2 \left\| \boldsymbol{v}^{(t+1)} - \boldsymbol{w}^{(t)} \right\|^2$.

	Since \( \boldsymbol{v}^{(t+1)} = \boldsymbol{w}^{(t)} - \lambda \boldsymbol{\Theta}^{(t)} \), where \( \boldsymbol{\Theta}^{(t)} \) is the average of local gradients, we have:
	\begin{equation}
	\left\| \boldsymbol{v}^{(t+1)} - \boldsymbol{w}^{(t)} \right\|^2 = \lambda^2 \left\| \boldsymbol{\Theta}^{(t)} \right\|^2 \leq \lambda^2 G^2,
	\end{equation}
	
	Therefore,
	\begin{equation}
	\beta = 2 L^2 \lambda^2 G^2.
	\end{equation}

	Returning to the inner product bound:
	\begin{equation}
	\begin{aligned}
	&  \left| \langle \nabla F(\boldsymbol{v}^{(t+1)}),\ \boldsymbol{w}^{(t+1)} - \boldsymbol{v}^{(t+1)} \rangle \right| \leq \\
	&\frac{1}{2\varepsilon} \left\| \nabla F(\boldsymbol{v}^{(t+1)}) \right\|^2  + \frac{\varepsilon}{2} \left\| \boldsymbol{w}^{(t+1)} - \boldsymbol{v}^{(t+1)} \right\|^2.
	\end{aligned}
	\end{equation}
	
	Choose \( \varepsilon = \alpha \). Then:
	\begin{equation}
	\begin{aligned}
	\frac{1}{2\varepsilon} \left\| \nabla F(\boldsymbol{v}^{(t+1)}) \right\|^2 &\leq \frac{1}{2\alpha} \left( \alpha \left\| \nabla F(\boldsymbol{w}^{(t)}) \right\|^2 + \beta \right) \\
	&= \frac{1}{2} \left\| \nabla F(\boldsymbol{w}^{(t)}) \right\|^2 + \frac{\beta}{2\alpha}.
	\end{aligned}
	\end{equation}
	
	Thus:
	\begin{equation}
	\begin{aligned}
	&\left| \langle \nabla F(\boldsymbol{v}^{(t+1)}),\ \boldsymbol{w}^{(t+1)} - \boldsymbol{v}^{(t+1)} \rangle \right| \leq\\ 
	&\frac{1}{2} \left\| \nabla F(\boldsymbol{w}^{(t)}) \right\|^2 + \frac{\beta}{2\alpha} + \frac{\alpha}{2} \left\| \boldsymbol{w}^{(t+1)} - \boldsymbol{v}^{(t+1)} \right\|^2.
	\end{aligned}
	\end{equation}
	
	According to the relationship 
	\begin{align}
	\begin{aligned}
	&\langle \nabla F(\boldsymbol{v}^{(t+1)}), \boldsymbol{w}^{(t+1)} - \boldsymbol{v}^{(t+1)} \rangle \leq \\
	&\left| \langle \nabla F(\boldsymbol{v}^{(t+1)}), \boldsymbol{w}^{(t+1)} - \boldsymbol{v}^{(t+1)} \rangle \right|,
	\end{aligned}
	\end{align}
	we obtain
	\begin{equation}
	\begin{aligned}
	&\langle \nabla F(\boldsymbol{v}^{(t+1)}),\ \boldsymbol{w}^{(t+1)} - \boldsymbol{v}^{(t+1)} \rangle  \leq \\
	&\frac{1}{2} \left\| \nabla F(\boldsymbol{w}^{(t)}) \right\|^2 + \frac{\beta}{2\alpha} + \frac{\alpha}{2} \left\| \boldsymbol{w}^{(t+1)} - \boldsymbol{v}^{(t+1)} \right\|^2.
	\end{aligned}
	\end{equation}
	Thus the Lemma \ref{lem:4} is proved.
	\subsection{Proof of Theorem 2}
	If \(\rho_{n}^{(t)*} = 0\), the constraint \(\alpha_{n}^{(t)*} \leq 1\) is not tight. The optimal $\alpha_{n}^{(t)*}$ is given by the analytical expression derived from the derivative condition:
	\begin{equation}
	\label{eq:85}
	\alpha_{n}^{(t)*} =\left(\frac{\sqrt{\eta^{(t)}}|h_n^{(t)}|}{\sqrt{P_n^{\mathrm{max}}}\left(|h_n^{(t)}|^2+\gamma_n^*\eta^{(t)}\right)}\right)^2.
	\end{equation}
	This solution must be verified against the constraint \(\alpha_{n}^{(t)*} \leq 1\): If the analytical solution satisfies \(\alpha_{n}^{(t)*} \leq 1\), the assumption that \(\rho_{n}^{(t)*} = 0\) holds, and the solution is directly given by (\ref{eq:85}). 
	If the analytical solution violates the constraint, i.e.,  \(\alpha_{n}^{(t)*} > 1\), the assumption \(\rho_{n}^{(t)*} = 0\) is invalid, and \(\rho_{n}^{(t)*}\) must be adjusted to ensure compliance with the constraint.
	
	If  \(\rho_{n}^{(t)*} > 0\) and the analytical solution exceeds the upper limit, i.e.,\(\left(\frac{\sqrt{\eta^{(t)}}|h_n^{(t)}|}{\sqrt{P_n^{\mathrm{max}}}\left(|h_n^{(t)}|^2+\gamma_n^*\eta^{(t)}\right)}\right)^2>1\), the constraint \(\alpha_{n}^{(t)*} \leq 1\) becomes active. By the complementary slackness condition: \(\rho_{n}^{(t)*} \left( \alpha_{n}^{(t)*} - 1 \right) = 0\),
	\(\rho_{n}^{(t)*}\) must be strictly positive (\(\rho_{n}^{(t)*} > 0\)) to enforce the constraint. Consequently,  \(\alpha_{n}^{(t)*}\) is directly set to the upper limit: $\alpha_{n}^{(t)*} = 1$.

	The role of \(\rho_{n}^{(t)*}\) is to enforce the constraint \(\alpha_{n}^{(t)*} \leq 1\). However, the final value of \(\alpha_{n}^{(t)*}\) is independent of the exact value of \(\rho_{n}^{(t)*}\). Therefore, \(\alpha_{n}^{(t)*}\) can be expressed as:
	\begin{equation}
	\alpha_{n}^{(t)*} = \min\left\{\left(\frac{\sqrt{\eta^{(t)}}|h_n^{(t)}|}{\sqrt{P_n^{\mathrm{max}}}\left(|h_n^{(t)}|^2+\gamma_n^*\eta^{(t)}\right)}\right)^2, 1\right\}.
	\end{equation}

	If \(\gamma_n^* > 0\), the complementary slackness condition $\gamma_n^* \left( \sum_{t=0}^{T-1} \alpha_{n}^{(t)*}  - \frac{T\bar{P}_n}{P_n^{\mathrm{max}}} \right) = 0$ ensures that the average power constraint is tight, i.e., $\sum_{t=0}^{T-1} \alpha_{n}^{(t)*}  = \frac{T\bar{P}_n}{P_n^{\mathrm{max}}}$.
	In this case, the value of \(\gamma_n^*\) can be determined using a bisection search method to enforce the equality.
	
	If \(\gamma_n^* = 0\), two scenarios arise: If \(\sum_{t=0}^{T-1} \alpha_{n}^{(t)*}  \leq \frac{T\bar{P}_n}{P_n^{\mathrm{max}}}\), the average power constraint is naturally satisfied, and no adjustment to \(\gamma_n^*\) is needed. If \(\sum_{t=0}^{T-1} \alpha_{n}^{(t)*}  > \frac{T\bar{P}_n}{P_n^{\mathrm{max}}}\), the average power constraint is violated, which contradicts the assumption \(\gamma_n^* = 0\). Therefore, in this case, \(\gamma_n^*\) must be strictly positive to tighten the constraint. 
	
	In summary, the optimal solution of $\alpha_{n}^{(t)*}$ is:
	\begin{equation}
	\left.\alpha_{n}^{(t)*}=\begin{cases}\min\left\{\frac{\eta^{(t)}}{P_n^{\mathrm{max}}|h_n^{(t)}|^2},1\right\},\\[2ex]\text{if }\sum_{t=0}^{T-1}\min\left\{\frac{\eta^{(t)}}{P_n^{\mathrm{max}}|h_n^{(t)}|^2},1\right\}< \frac{T\bar{P}_n}{P_n^{\mathrm{max}}},\\[2ex]\min\left\{\left(\frac{\sqrt{\eta^{(t)}}|h_n^{(t)}|}{\sqrt{P_n^{\mathrm{max}}}\left(|h_n^{(t)}|^2+\gamma_n^*\eta^{(t)}\right)}\right)^2,1\right\},\\[2ex]\text{if }\sum_{t=0}^{T-1}\min\left\{\frac{\eta^{(t)}}{P_n^{\mathrm{max}}|h_n^{(t)}|^2},1\right\}\geq\frac{T\bar{P}_n}{P_n^{\mathrm{max}}},\end{cases}\right.
	\end{equation}
	where $\gamma_n^*$ can be found to ensure the average power constraint $\sum_{t=0}^{T-1}\alpha_{n}^{(t)*}= \frac{T\bar{P}_n}{P_n^{\mathrm{max}}}$ via the bisection search method.

}

\bibliographystyle{IEEEtran}
\bibliography{IEEEabrv,myref}

\begin{thebibliography}{10}
\providecommand{\url}[1]{#1}
\csname url@samestyle\endcsname
\providecommand{\newblock}{\relax}
\providecommand{\bibinfo}[2]{#2}
\providecommand{\BIBentrySTDinterwordspacing}{\spaceskip=0pt\relax}
\providecommand{\BIBentryALTinterwordstretchfactor}{4}
\providecommand{\BIBentryALTinterwordspacing}{\spaceskip=\fontdimen2\font plus
\BIBentryALTinterwordstretchfactor\fontdimen3\font minus
  \fontdimen4\font\relax}
\providecommand{\BIBforeignlanguage}[2]{{%
\expandafter\ifx\csname l@#1\endcsname\relax
\typeout{** WARNING: IEEEtran.bst: No hyphenation pattern has been}%
\typeout{** loaded for the language `#1'. Using the pattern for}%
\typeout{** the default language instead.}%
\else
\language=\csname l@#1\endcsname
\fi
#2}}
\providecommand{\BIBdecl}{\relax}
\BIBdecl

\bibitem{baker2023artificial}
S.~Baker and W.~Xiang, ``Artificial intelligence of things for smarter
  healthcare: A survey of advancements, challenges, and opportunities,''
  \emph{IEEE Communications Surveys \& Tutorials}, vol.~25, no.~2, pp.
  1261--1293, 2023.

\bibitem{tallat2023navigating}
R.~Tallat, A.~Hawbani, X.~Wang, A.~Al-Dubai, L.~Zhao, Z.~Liu, G.~Min, A.~Y.
  Zomaya, and S.~H. Alsamhi, ``Navigating industry 5.0: A survey of key
  enabling technologies, trends, challenges, and opportunities,'' \emph{IEEE
  Communications Surveys \& Tutorials}, 2023.

\bibitem{abkenar2022survey}
F.~S. Abkenar, P.~Ramezani, S.~Iranmanesh, S.~Murali, D.~Chulerttiyawong,
  X.~Wan, A.~Jamalipour, and R.~Raad, ``A survey on mobility of edge computing
  networks in iot: State-of-the-art, architectures, and challenges,''
  \emph{IEEE Communications Surveys \& Tutorials}, vol.~24, no.~4, pp.
  2329--2365, 2022.

\bibitem{liu2021machine}
B.~Liu, M.~Ding, S.~Shaham, W.~Rahayu, F.~Farokhi, and Z.~Lin, ``When machine
  learning meets privacy: A survey and outlook,'' \emph{ACM Computing Surveys
  (CSUR)}, vol.~54, no.~2, pp. 1--36, 2021.

\bibitem{mcmahan2017communication}
B.~McMahan, E.~Moore, D.~Ramage, S.~Hampson, and B.~A. y~Arcas,
  ``Communication-efficient learning of deep networks from decentralized
  data,'' in \emph{Artificial intelligence and statistics}.\hskip 1em plus
  0.5em minus 0.4em\relax PMLR, 2017, pp. 1273--1282.

\bibitem{karimireddy2020scaffold}
S.~P. Karimireddy, S.~Kale, M.~Mohri, S.~Reddi, S.~Stich, and A.~T. Suresh,
  ``Scaffold: Stochastic controlled averaging for federated learning,'' in
  \emph{International conference on machine learning}.\hskip 1em plus 0.5em
  minus 0.4em\relax PMLR, 2020, pp. 5132--5143.

\bibitem{vithana2023private}
S.~Vithana and S.~Ulukus, ``Private read update write (pruw) in federated
  submodel learning (fsl): Communication efficient schemes with and without
  sparsification,'' \emph{IEEE Transactions on Information theory}, 2023.

\bibitem{tang2022gossipfl}
Z.~Tang, S.~Shi, B.~Li, and X.~Chu, ``Gossipfl: A decentralized federated
  learning framework with sparsified and adaptive communication,'' \emph{IEEE
  Transactions on Parallel and Distributed Systems}, vol.~34, no.~3, pp.
  909--922, 2022.

\bibitem{liu2022hierarchical}
L.~Liu, J.~Zhang, S.~Song, and K.~B. Letaief, ``Hierarchical federated learning
  with quantization: Convergence analysis and system design,'' \emph{IEEE
  Transactions on Wireless Communications}, vol.~22, no.~1, pp. 2--18, 2022.

\bibitem{chen2022energy}
R.~Chen, L.~Li, K.~Xue, C.~Zhang, M.~Pan, and Y.~Fang, ``Energy efficient
  federated learning over heterogeneous mobile devices via joint design of
  weight quantization and wireless transmission,'' \emph{IEEE Transactions on
  Mobile Computing}, vol.~22, no.~12, pp. 7451--7465, 2022.

\bibitem{xu2020client}
J.~Xu and H.~Wang, ``Client selection and bandwidth allocation in wireless
  federated learning networks: A long-term perspective,'' \emph{IEEE
  Transactions on Wireless Communications}, vol.~20, no.~2, pp. 1188--1200,
  2020.

\bibitem{nishio2019client}
T.~Nishio and R.~Yonetani, ``Client selection for federated learning with
  heterogeneous resources in mobile edge,'' in \emph{ICC 2019-2019 IEEE
  international conference on communications (ICC)}.\hskip 1em plus 0.5em minus
  0.4em\relax IEEE, 2019, pp. 1--7.

\bibitem{zhu2019broadband}
G.~Zhu, Y.~Wang, and K.~Huang, ``Broadband analog aggregation for low-latency
  federated edge learning,'' \emph{IEEE Transactions on Wireless
  Communications}, vol.~19, no.~1, pp. 491--506, 2019.

\bibitem{yang2020federated}
K.~Yang, T.~Jiang, Y.~Shi, and Z.~Ding, ``Federated learning via over-the-air
  computation,'' \emph{IEEE transactions on wireless communications}, vol.~19,
  no.~3, pp. 2022--2035, 2020.

\bibitem{zhang2021gradient}
N.~Zhang and M.~Tao, ``Gradient statistics aware power control for over-the-air
  federated learning,'' \emph{IEEE Transactions on Wireless Communications},
  vol.~20, no.~8, pp. 5115--5128, 2021.

\bibitem{cao2022transmission}
X.~Cao, G.~Zhu, J.~Xu, and S.~Cui, ``Transmission power control for
  over-the-air federated averaging at network edge,'' \emph{IEEE Journal on
  Selected Areas in Communications}, vol.~40, no.~5, pp. 1571--1586, 2022.

\bibitem{ni2022integrating}
W.~Ni, Y.~Liu, Z.~Yang, H.~Tian, and X.~Shen, ``Integrating over-the-air
  federated learning and non-orthogonal multiple access: What role can ris
  play?'' \emph{IEEE Transactions on Wireless Communications}, vol.~21, no.~12,
  pp. 10\,083--10\,099, 2022.

\bibitem{shao2021federated}
Y.~Shao, D.~G{\"u}nd{\"u}z, and S.~C. Liew, ``Federated edge learning with
  misaligned over-the-air computation,'' \emph{IEEE Transactions on Wireless
  Communications}, vol.~21, no.~6, pp. 3951--3964, 2021.

\bibitem{sun2021dynamic}
Y.~Sun, S.~Zhou, Z.~Niu, and D.~G{\"u}nd{\"u}z, ``Dynamic scheduling for
  over-the-air federated edge learning with energy constraints,'' \emph{IEEE
  Journal on Selected Areas in Communications}, vol.~40, no.~1, pp. 227--242,
  2021.

\bibitem{xu2021learning}
C.~Xu, S.~Liu, Z.~Yang, Y.~Huang, and K.-K. Wong, ``Learning rate optimization
  for federated learning exploiting over-the-air computation,'' \emph{IEEE
  Journal on Selected Areas in Communications}, vol.~39, no.~12, pp.
  3742--3756, 2021.

\bibitem{zou2022knowledge}
Y.~Zou, Z.~Wang, X.~Chen, H.~Zhou, and Y.~Zhou, ``Knowledge-guided learning for
  transceiver design in over-the-air federated learning,'' \emph{IEEE
  Transactions on Wireless Communications}, vol.~22, no.~1, pp. 270--285, 2022.

\bibitem{albaseer2023fair}
A.~M. Albaseer, M.~Abdallah, A.~Al-Fuqaha, A.~M. Seid, A.~Erbad, and O.~A.
  Dobre, ``Fair selection of edge nodes to participate in clustered federated
  multitask learning,'' \emph{IEEE Transactions on Network and Service
  Management}, vol.~20, no.~2, pp. 1502--1516, 2023.

\bibitem{li2021ditto}
T.~Li, S.~Hu, A.~Beirami, and V.~Smith, ``Ditto: Fair and robust federated
  learning through personalization,'' in \emph{International conference on
  machine learning}.\hskip 1em plus 0.5em minus 0.4em\relax PMLR, 2021, pp.
  6357--6368.

\bibitem{liu2021fedcpf}
S.~Liu, J.~Yu, X.~Deng, and S.~Wan, ``Fedcpf: An efficient-communication
  federated learning approach for vehicular edge computing in 6g communication
  networks,'' \emph{IEEE Transactions on Intelligent Transportation Systems},
  vol.~23, no.~2, pp. 1616--1629, 2021.

\bibitem{wu2020accelerating}
W.~Wu, L.~He, W.~Lin, and R.~Mao, ``Accelerating federated learning over
  reliability-agnostic clients in mobile edge computing systems,'' \emph{IEEE
  Transactions on Parallel and Distributed Systems}, vol.~32, no.~7, pp.
  1539--1551, 2020.

\bibitem{shi2024analysis}
F.~Shi, W.~Lin, X.~Wang, K.~Li, and A.~Y. Zomaya, ``The analysis and
  optimization of volatile clients in over-the-air federated learning,''
  \emph{IEEE Transactions on Mobile Computing}, 2024.

\bibitem{guo2022joint}
W.~Guo, R.~Li, C.~Huang, X.~Qin, K.~Shen, and W.~Zhang, ``Joint device
  selection and power control for wireless federated learning,'' \emph{IEEE
  Journal on Selected Areas in Communications}, vol.~40, no.~8, pp. 2395--2410,
  2022.

\bibitem{zhao2018federated}
Y.~Zhao, M.~Li, L.~Lai, N.~Suda, D.~Civin, and V.~Chandra, ``Federated learning
  with non-iid data,'' \emph{arXiv preprint arXiv:1806.00582}, 2018.

\bibitem{wu2023joint}
B.~Wu, F.~Fang, and X.~Wang, ``Joint age-based client selection and resource
  allocation for communication-efficient federated learning over noma
  networks,'' \emph{IEEE Transactions on Communications}, 2023.

\bibitem{wang2024age}
K.~Wang, Z.~Ding, D.~K. So, and Z.~Ding, ``Age-of-information minimization in
  federated learning based networks with non-iid dataset,'' \emph{IEEE
  Transactions on Wireless Communications}, 2024.

\bibitem{ma2023channel}
M.~Ma, V.~W. Wong, and R.~Schober, ``Channel-aware joint aoi and diversity
  optimization for client scheduling in federated learning with non-iid
  datasets,'' \emph{IEEE Transactions on Wireless Communications}, 2023.

\bibitem{dong2024age}
L.~Dong, Y.~Zhou, L.~Liu, Y.~Qi, and Y.~Zhang, ``Age of information based
  client selection for wireless federated learning with diversified learning
  capabilities,'' \emph{IEEE Transactions on Mobile Computing}, 2024.

\bibitem{Li2020On}
\BIBentryALTinterwordspacing
X.~Li, K.~Huang, W.~Yang, S.~Wang, and Z.~Zhang, ``On the convergence of fedavg
  on non-iid data,'' in \emph{International Conference on Learning
  Representations}, 2020. [Online]. Available:
  \url{https://openreview.net/forum?id=HJxNAnVtDS}
\BIBentrySTDinterwordspacing

\bibitem{kadota2018scheduling}
I.~Kadota, A.~Sinha, E.~Uysal-Biyikoglu, R.~Singh, and E.~Modiano, ``Scheduling
  policies for minimizing age of information in broadcast wireless networks,''
  \emph{IEEE/ACM Transactions on Networking}, vol.~26, no.~6, pp. 2637--2650,
  2018.

\bibitem{cao2020optimized}
X.~Cao, G.~Zhu, J.~Xu, and K.~Huang, ``Optimized power control for over-the-air
  computation in fading channels,'' \emph{IEEE Transactions on Wireless
  Communications}, vol.~19, no.~11, pp. 7498--7513, 2020.

\bibitem{deng2021auction}
Y.~Deng, F.~Lyu, J.~Ren, H.~Wu, Y.~Zhou, Y.~Zhang, and X.~Shen, ``Auction:
  Automated and quality-aware client selection framework for efficient
  federated learning,'' \emph{IEEE Transactions on Parallel and Distributed
  Systems}, vol.~33, no.~8, pp. 1996--2009, 2021.

\bibitem{li2019wirelessly}
X.~Li, G.~Zhu, Y.~Gong, and K.~Huang, ``Wirelessly powered data aggregation for
  iot via over-the-air function computation: Beamforming and power control,''
  \emph{IEEE Transactions on Wireless Communications}, vol.~18, no.~7, pp.
  3437--3452, 2019.

\end{thebibliography}

\end{document}